\documentclass[12pt]{article}

\usepackage{latexsym}
\usepackage{graphicx}
\usepackage{amsfonts}
\usepackage{amsmath, amsthm, amssymb}
\usepackage{fullpage}
\usepackage{setspace}
\usepackage{longtable}
\usepackage{natbib}
\bibpunct{(}{)}{;}{a}{}{,}
\usepackage{dcolumn}
\newcolumntype{.}{D{.}{.}{-1}}
\newcolumntype{d}[1]{D{.}{.}{#1}}
\usepackage{bm}
\usepackage{threeparttable}
\usepackage{booktabs}
\usepackage{enumerate}
\usepackage{bbm}
\usepackage[top=1in, left=1in, right=1in,
bottom=1in]{geometry} 
\usepackage{appendix}

\newcommand{\reals}{\mathbb{R}}
\newcommand{\Var}{{\rm Var}}
\newcommand{\E}{{\rm E}}
\newcommand{\Prob}{{\rm P}}

\newtheorem{corollary}{Corollary}

\newtheorem{proposition}{Proposition}

\theoremstyle{definition}

\newtheorem{assumption}{Assumption}

\theoremstyle{remark}

\usepackage{caption}
\usepackage{subcaption}

\begin{document}
\pagestyle{plain}

\newcommand{\blind}{0}

\newcommand{\tit}{\large{Instrumental Variables Estimation in Cluster Randomized Trials with Noncompliance: A Study of A Biometric Smartcard Payment System in India}}

\if0\blind

{\title{\tit\thanks{We thank participants at the 2018 Atlantic Causal Inference Conference, Penn Causal Inference Seminar, and the University of Florida 2019 Winter Workshop for helpful comments.}}
\author{Hyunseung Kang\thanks{Assistant Professor, University of Wisconsin, Madison, Email: hyunseung@stat.wisc.edu}
\and Luke Keele\thanks{Associate Professor, University of Pennsylvania, Email:
      luke.keele@pennmedicine.upenn.edu}}

\date{\today}

\maketitle
}\fi

\if1\blind
\title{\bf \tit}
\maketitle
\fi

\maketitle

\thispagestyle{empty}

\begin{abstract}
Many policy evaluations occur in settings where treatment is randomized at the cluster level, and there is treatment noncompliance within each cluster. For example, villages might be assigned to treatment and control, but residents in each village may choose to comply or not with their assigned treatment status. When noncompliance is present, the instrumental variables framework can be used to identify and estimate causal effects. While a large literature exists on instrumental variables estimation methods, relatively little work has been focused on settings with clustered treatments. Here, we review extant methods for instrumental variable estimation in clustered designs and derive both the finite and asymptotic properties of these estimators. We prove that the properties of current estimators depend on unrealistic assumptions. We then develop a new IV estimation method for cluster randomized trials and study its formal properties. We prove that our IV estimator allows for possible treatment effect heterogeneity that is correlated with cluster size and is robust to low compliance rates within clusters. We evaluate these methods using simulations and apply them to data from a randomized intervention in India.
\end{abstract}

\doublespacing
\clearpage

\section{Introduction}
In many policy settings, randomized trials are used to evaluate policy. Randomized trials allow analysts to rule out that causal effect estimates are confounded with pretreatment differences or selection biases amongst subjects.  In education and public health applications, investigators often use clustered randomized trials (CRTs), where treatments are applied to groups of individuals rather than individuals.  While clustered treatment assignment reduces power, it allows for arbitrary patterns of treatment spillover for individuals within the same cluster \citep{imbens2008recent}. This design feature is critical in settings where interactions between individuals within clusters are difficult to prevent.

One prototypical example of a CRT in a policy setting occurred in the state of Andhra Pradesh, India. The policy goal was to evaluate the effectiveness of using a biometric payment system to deliver social welfare payments to recipients in India \citep{muralidharan2016building}. The smartcard payment system used a network of locally hired, bank-employed staff to biometrically authenticate beneficiaries and make cash payments in villages. Biometric systems are designed to reduce the time it takes for payments to reach payees and reduce the chance that payments are siphoned off to someone other than the payee. The intervention was randomly assigned at the village level. Out of 157 villages, 112 were assigned to the intervention. Assignment at the village level prevented treatment spillovers that would have been difficult to prevent if assigned at the individual level.

For villages assigned to the control condition, smartcard payments were not available. For those in treated villages, recipients first had to enroll in the smartcard program by submitting biometric data (typically all ten fingerprints) and digital photographs. Beneficiaries were then issued a physical smartcard that included their photograph and an embedded electronic chip storing biographic and biometric data. The card was also linked to newly created bank accounts. Government officials conducted enrollment campaigns to collect the biometric data and issue the cards. The government contracted with banks to manage payments, and these banks in turn contracted with customer service providers (CSPs) to manage the accounts and travel to villages to deliver payments. The system was used to deliver payments from two large welfare programs: the National Rural Employment Guarantee Scheme (NREGS), and Social Security Pensions (SSP).  The first is a work-fare program that guarantees every rural household 100 days of paid employment each year \citep{dutta2010small}. SSP complements the first program by providing income support to those who are unable to work. Beneficiaries used the smartcards to collect payments from CSPs by inserting them into a point-of-service device. The device reads the card and retrieves account details. Payees were prompted for one of ten fingers, chosen at random, to be scanned. The device compares this scan with the records on the card, and authorizes a transaction if they match. Once authorized the amount of cash requested is disbursed, and the device prints out a receipt.
 
\citet{muralidharan2016building} conducted the original evaluation of the smartcard payment system but did not account for noncompliance in the analysis. In the smartcard evaluation, as is true in many CRTs with human subjects, compliance with the treatment assignment varied. Beneficiaries were designated as compliant if they used the smartcard one or more times to collect payments. In some villages, 90\% or more of the beneficiaries complied with the treatment, while in many villages less 10\% of the recipients complied with their assigned treatment assignment. 

While noncompliance is common in these designs, the statistical literature contains relatively little guidance on the analysis of CRTs with noncompliance. For example, two widely used textbooks on CRTs do not mention noncompliance at all \citep{hayes2009cluster, donner2000design}. Noncompliance in CRTs can be analyzed within the instrumental variables (IV) framework. IV methods allow investigators to estimate the average causal treatment effect among the subpopulation that complied with their assigned treatment \citep{Angrist:1996}. \citet{jo2008cluster} and \citet{schochet2013estimators} both extended the IV framework to CRTs.

In this paper, we study and develop methods of estimation and inference for CRTs within noncompliance. We review the two most commonly used methods of estimation and inference: one based on cluster level averages, and a second using individual unit level data. We demonstrate that both methods impose implicit assumptions on the heterogeneity of effects across clusters and the size of each cluster. We then present an alternative method of estimation that relaxes these assumptions and provides accurate finite sample inferences. Using a simulation study, we compare these different methods of estimation and inference. We conclude with an analysis of the smartcard trial which motivated the study.

\section{Review of Clustered Randomized Trials and Noncompliance}
\subsection{Notation}

Suppose there are $J$ total clusters, indexed by  $j=1,\ldots,J$. For each cluster $j$, there are $n_{j}$ units, indexed by $i=1, \ldots, n_{j}$ and we have $n = \sum_{j=1}^{J} n_j$ total units. A fixed $m$ number of clusters are assigned treatment and the other $J - m$ clusters are assigned control. The treatment assignment is uniform within each cluster; if cluster $j$ is assigned treatment, all $n_{j}$ individuals in cluster $j$ are assigned treatment. Let $Z_{j} = 1$ indicate that cluster $j$ received treatment and $Z_{j} = 0$ indicates that cluster $j$ received control. Let $\mathbf{Z} = (Z_1,\ldots,Z_J)$ be a vector of $Z_{j}$s. Also, let $\mathbf{z} = (z_{1},\ldots,z_{J})$ be one possible value of treatment assignment from a set of treatment assignments $\mathcal{Z} = \{\mathbf{z} \in \{0,1\}^{J}\mid \sum_{j=1}^{J} Z_{j} = m\}$. Let $D_{ji} = 1$ indicate that unit $i$ in cluster $j$ actually took the treatment and $D_{ji} = 0$ indicate that the unit actually took the control. Let $Y_{ji}$ denote the observed outcome for individual $i$ in cluster $j$. We can also use $Z_{ji}$ to denote the treatment assignment of individual $i$ in cluster $j$, but because treatment assignment is uniform within each cluster, we use this notation infrequently.

We use the potential outcomes approach to define causal effects \citep{Neyman:1923a,rubin_estimating_1974}. Let $D_{ji}^{(1)}$ and $D_{ji}^{(0)}$ indicate the potential treatment receipt of individual $i$ in cluster $j$ if cluster $j$ received treatment $Z_{j} = 1$ or control $Z_{j} = 0$, respectively.  Let $Y_{ji}^{(1,d_{ji})}$ and $Y_{ji}^{(0,d_{ji})}$ indicate the potential outcomes of individual $i$ in cluster $j$ if cluster $j$ received treatment $Z_{j} = 1$ or control $Z_{j} =0$, respectively, and each individual's treatment receipt was set to $d_{ji} \in \{0,1\}$. Let $Y_{ji}^{\left(1,D_{ji}^{(1)}\right)}$ and $Y_{ji}^{\left(0,D_{ji}^{(0)}\right)}$ indicate the potential outcomes of individual $i$ in cluster $j$ if cluster $j$ received treatment $Z_{j} = 1$ or control $Z_{j} = 0$, respectively and each individual took on their ``natural'' potential treatment receipt. Because only one of the two potential outcomes can be observed for a treatment assignment $Z_j$, the relationship between the potential outcomes and observed values is 
\[
Y_{ji} = Z_{j} Y_{ji}^{\left(1,D_{ji}^{(1)}\right)} + (1 - Z_j) Y_{ji}^{\left(0,D_{ji}^{(0)}\right)}, \quad{} D_{ji} = Z_{j} D_{ji}^{(1)} + (1 -Z_j) D_{ji}^{(0)}
\]
Let $\mathcal{F} = \left\{Y_{ji}^{\left(z,D_{ji}^{(z)}\right)},D_{ji}^{(z)}, z=0,1,j=1,\ldots,J, i =1,\ldots,n_{j} \right\}$ be the set containing all values of the potential outcomes. We note that the potential outcomes notation implicitly assumes that there is no interference and the stable unit treatment value assumption (SUTVA) holds \citep{Rubin:1986}. This assumption is commonly assumed in CRTs; see \citet{frangakis_clustered_2002}, \citet{small_randomization_2008}, \citet{jo_intention_2008}, \citet{imai_essential_2009}, \citet{schochet_estimation_2011} and \citet{middleton_unbiased_2015} for examples.

We also define a set of summary statistics for each cluster $j$. For each cluster $j$, let $Y_{j} = \sum_{i=1}^{n_j} Y_{ji}$ and $D_{j} = \sum_{i=1}^{n_j} {D}_{ij}$ indicate the sums of individual outcomes and treatment receipts, respectively. Similarly, for each cluster $j$, let $\overline{Y}_j = Y_j / n_j$ and $\overline{D}_{j} = D_j / n_j$ indicate the average of individual outcomes and treatment receipts, respectively. For each cluster $j$ and treatment indicator $z \in \{0,1\}$, let $Y_{j}^{\left(z,D_{j}^{(z)}\right)} = \sum_{i=1}^{n_j} Y_{ji}^{\left(z,D_{ji}^{(z)} \right)}$ and $D_{j}^{(z)} = \sum_{i=1}^{n_j} D_{ji}^{(z)}$ indicate the sums of potential outcomes and treatment receipts, respectively, when cluster $j$ is assigned treatment $z$. Similarly, let $\overline{Y}_{j}^{\left(z,D_{j}^{(z)}\right)} = Y_{j}^{\left(z,D_{j}^{(z)}\right)}  / n_j$ and $\overline{D}_{j}^{(z)} = D_{j}^{(z)} / n_j$ indicate the averages of potential outcomes and potential treatment receipts, respectively, when cluster $j$ is assigned treatment $z$.

\subsection{Assumptions} 
\label{sec:A}
Next, we review the assumptions underlying CRTs with noncompliance. We begin by discussing assumptions specific to CRTs; see \citet{hayes2009cluster} and \citet{donner2000design} for details. We assume that each cluster is randomly assigned to treatment $Z_{j} = 1$  or control $Z_{j}=0$ and that the probability of receiving either one is non-zero.
\begin{assumption}[Cluster Randomization] \label{a1}Given $J$ clusters, $m$ clusters ($0< m < J$) are randomly assigned treatment and the other $J - m$ clusters are assigned control: 
\[
\Prob(\mathbf{Z} = \mathbf{z} \mid \mathcal{F}, \mathcal{Z}) = \Prob(\mathbf{Z} = \mathbf{z} \mid \mathcal{Z}) \text{ and } \Prob(\mathbf{Z} = \mathbf{z} \mid \mathcal{Z}) = \frac{1}{{J \choose m}}
\]
\end{assumption}
Assumption \ref{a1} holds in a typical CRT, since clusters are randomly assigned to either treatment or control. Under Assumption \ref{a1}, we can identify and unbiasedly estimate the following average causal effects:
\begin{align*}
\mu_Y &= \frac{1}{n} \sum_{j=1}^{J} \sum_{i=1}^{n_j} Y_{ji}^{\left(1,D_{ji}^{(1)}\right)} - Y_{ji}^{\left(0,D_{ji}^{(0)}\right)}, \quad{} \mu_D = \frac{1}{n} \sum_{j=1}^{J} \sum_{i=1}^{n_j} D_{ji}^{(1)} -  D_{ji}^{(0)}.
\end{align*}
The first effect $\mu_Y$ is the average causal effect of treatment assignment on the outcome and the second effect $\mu_D$ is the average causal effect of the treatment assignment on treatment receipt.
$\mu_Y$ is often referred to as the intent-to-treat (ITT) effect and $\mu_D$ is often referred to as the compliance rate. Unbiased estimators for $\mu_Y$ and $\mu_D$, denoted as $\hat{\mu}_Y$ and $\hat{\mu}_D$, respectively, are difference-in-means estimators of sums between treated and control clusters:
\begin{align*}
\hat{\mu}_Y &= \frac{1}{n} \left(\frac{J}{m} \sum_{j=1}^{J} Z_j Y_j - \frac{J}{J -m} \sum_{j=1}^{J} (1 - Z_j) Y_j \right), \quad{} \hat{\mu}_D = \frac{1}{n} \left(\frac{J}{m} \sum_{j=1}^{J} Z_j D_j - \frac{J}{J -m} \sum_{j=1}^{J} (1 - Z_j) D_j \right)
\end{align*}
where $\E[\hat{\mu}_Y  \mid \mathcal{F}, \mathcal{Z}] = \mu_Y$ and $\E[\hat{\mu}_D \mid \mathcal{F}, \mathcal{Z}] = \mu_D$ under Assumption \ref{a1}.

Next, we review assumptions that are specific to noncompliance; see \citet{Imbens:1994}, \citet{Angrist:1996}, \citet{hernan_instruments_2006}, and \citet{baiocchi_instrumental_2014} for a detailed discussion of these assumptions. We outline these assumptions in the context of the smartcard study in our motivating example. We start by assuming that the compliance rate is non-zero and no individual systematically defies his/her treatment assignment, also referred to as the monotonicity assumption \citep{Imbens:1994}.
\begin{assumption}[Non-Zero Causal Effect] \label{a2} The treatment assignment $z_{j}$, on average, causes a changes in treatment receipt $D_{ij}$, i.e. $\mu_D \neq 0$.
\end{assumption}
\begin{assumption}[Monotonicity]  \label{a3} There is no individual who systematically defies the treatment assignment, i.e. $D_{ji}^{(0)} \leq D_{ji}^{(1)}$ for all $ij$.
\end{assumption}
In the context of the smartcard experiment, Assumption \ref{a2} states that on average, individuals signed up for the biometric system when their village offered it. If Assumption \ref{a1} holds, Assumption \ref{a2} can be verified from data by using the estimator $\hat{\mu}_{D}$.  Assumption \ref{a3} can be interpreted by categorizing the study participant into four types, always-takers, never-takers, compliers, and defiers, based on his/her potential treatment receipts $D_{ji}^{(0)}$ and $D_{ji}^{(1)}$ \citep{Angrist:1996}. Always-takers are individuals where $D_{ji}^{(1)} = D_{ji}^{(0)} = 1$; these individuals would use the smartcard payment system irrespective of village treatment status. Never-takers are individuals where $D_{ji}^{(1)} = D_{ji}^{(0)} = 0$; these individuals would never use smartcard payments irrespective of the village treatment assignment. Compliers are individuals where $D_{ji}^{(1)} = 1$ and $D_{ji}^{(0)} = 0$; they would use smartcard payments only when their village is assigned to treatment. Defiers are individuals where $D_{ji}^{(1)} = 0$ and $D_{ji}^{(0)} = 1$; these individual would act contrary to assigned village treatment status. Assumption \ref{a3} states that there are no defiers in the study population. Assumption \ref{a3} can hold by design if units receiving the control cannot obtain the treatment, i.e. $D_{ji}^{(0)} = 0$; this is known as one-sided noncompliance in the literature. For example, in the smartcard intervention, noncompliance was one-sided, since the key technology for payment verification was unavailable in the control villages. 


Finally, we assume that conditional on treatment received, the initial treatment assignment has no impact on the outcome. This is commonly known as the exclusion restriction.
\begin{assumption}[Exclusion Restriction] \label{a4} Conditional on treatment receipt $d_{ji}$, the treatment assignment has no effect on the outcome, i.e. $Y_{ji}^{(1,d_{ji})} = Y_{ji}^{(0,d_{ji})} = Y_{ji}^{(d_{ji})}$ for all $d_{ji}$.
\end{assumption}
Typically, Assumption \ref{a4} is assessed based on substantive knowledge about the treatment and the outcome. For example, in the smartcard study, for Assumption \ref{a4} to hold, it must be the case that being assigned to smartcard payments has no direct effect on the outcome except through exposure to the intervention. If the outcome is time to payment (see Section \ref{sec:data} for details), being assigned to the smartcard payment system can only reduce payment delay through actual use of smartcard payments. This seems likely to hold, since the reduced payment times are directly facilitated by the technology that underlies the intervention.

\subsection{Complier Average Treatment Effect}
In CRTs with noncompliance, one causal estimand of interest is the complier average causal effect (CACE), which is the average effect of actually taking the treatment (i.e. when $D_{ji}$ is set to $1$) versus not taking the treatment (i.e. when $D_{ji}$ is set to $0$) among compliers. We denote the CACE as $\tau$:
\begin{equation} \label{eq:CATE}
 \tau = \frac{\sum_{j=1}^{J} \sum_{i=1}^{n_j} (Y_{ji}^{(1)} - Y_{ji}^{(0)}) I(D_{ji}^{(1)} = 1, D_{ji}^{(0)} = 0)}{\sum_{j=1}^{J} \sum_{i=1}^{n_j}  I(D_{ji}^{(1)} = 1, D_{ji}^{(0)} = 0)}.
\end{equation} 
The parameter $\tau$ is also referred to as a local causal effect because it only describes the causal effect among a subgroup of individuals in the population, specifically the compliers. In the smartcard study, $\tau$ is the effect of using smartcard payments among those recipient who were induced to use smartcards when exposed to the intervention. Prior literature has shown that under Assumptions \ref{a1}-\ref{a4}, $\tau$ is identified by taking the ratio of the ITT effect $\mu_Y$ with the compliance rate $\mu_D$, i.e. $\mu_Y / \mu_D$ \citep{Imbens:1994, Angrist:1996}.

In the context of a CRT, we can decompose $\tau$ to understand how CACE may vary across clusters. First, let $n_{{\rm CO},j} = \sum_{i=1}^{n_j}  I(D_{ji}^{(1)} = 1, D_{ji}^{(0)} = 0)$ be the number of compliers in cluster $j$, $n_{{\rm CO}} = \sum_{j=1}^{J}n_{{\rm CO},j}$ be the total number of compliers across all clusters, and $\tau_j$ be the CACE for cluster $j$, i.e.
\begin{equation*}
\tau_j = \frac{\sum_{i=1}^{n_j} (Y_{ji}^{(1)} - Y_{ji}^{(0)}) I(D_{ji}^{(1)} = 1, D_{ji}^{(0)} = 0)}{n_{{\rm CO},j}}
\end{equation*}
Using these terms, we can rewrite $\tau$ as a weighted average of cluster-specific CACE
\begin{equation}
\label{eq:tau_j}	
\tau = \sum_{j=1}^{J} w_j \tau_j, \quad{} w_j = \frac{n_{{\rm CO},j}}{ \sum_{j=1}^{J}  n_{{\rm CO},j}}.
\end{equation}
where the weights, denoted as $(w_1,\ldots,w_J)$ are functions of the number of compliers in each cluster. If some clusters have no compliers, i.e. clusters $j$ where $n_{{\rm CO},j} = 0$, we would decompose $\tau$ as sum over clusters $j$ with $n_{{\rm CO},j} > 0$. Finally, if the cluster specific CACEs $\tau_j$s are constant across clusters so that $\tau_1 = \tau_2 = \cdots = \tau_{J-1} = \tau_J$, the population CACE $\tau$ is equal to the cluster CACE $\tau_j$. 

Finally, we conclude the discussion of CACE by formalizing hypothesis testing for $\tau$. For any $\tau_0 \in \reals$, let $H_0$ denote the null hypothesis for $\tau$:
\begin{equation} \label{eq:hyp}
H_0:  \tau = \tau_0, \quad{} \tau_0 \in \reals
\end{equation}
The null hypothesis $H_0$ in equation \eqref{eq:hyp} is a composite null hypothesis because there are several values of $\mathcal{F}$ for which the null holds. In other words, $H_0$ is not a sharp null hypothesis \citep{Fisher:1935}; under a sharp null, we would be able to specify exactly the unobserved potential outcomes. In fact, the sharp null of no ITT effect, $Y_{i}^{\left(1,D_{i}^{(1)}\right)} = Y_{i}^{\left(0,D_{i}^{(0)}\right)}$ for all $i$ implies $H_0: \tau = 0$, but the converse is not necessarily true; there are other values of $\mathcal{F}$ that satisfies the null hypothesis $H_0: \tau = 0$. Next, we describe extant procedures for estimating the CACE in CRTs with noncompliance.
 
\section{Review of Extant Methods for Analysis of CRTs with Noncompliance}

The first method we review is outlined in \citet{Hansen:2008b} and \citet{schochet2013estimators}. This method aggregates individual observations at the cluster level and conducts a statistical analysis with the aggregate cluster-level quantities. 
The second method applies two-stage least squares (TSLS) to the unit level data and uses robust standard errors to account for within cluster correlations. A third approach, which we do not explore, is to implement IV methods using random effects models \citep{jo2008cluster,small2006random} because they typically rely on specifying additional distributional assumptions.

\subsection{The Cluster-Level Method}
\label{sec:clus}
Cluster-level methods analyzes CRTs with noncompliance at the cluster level. Specifically, the investigator takes averages or sums of individual outcomes and compliances within each cluster, and then treat the study as though it only $J$ effective subjects with $J$ cluster-level summaries \citep{Hansen:2008b,schochet2013estimators}. Formally, given cluster-level average outcomes $\overline{Y}_{j}$ and treatment receipts $\overline{D}_{j}$, we define the overall average outcomes and compliances by treatment status
\begin{align*}
\overline{Y}_T &= \frac{1}{m}\sum_{j=1}^J Z_j\overline{Y}_j, \quad{} \overline{Y}_C = \frac{1}{J-m}\sum_{j=1}^{J} (1 - Z_j) \overline{Y}_j \\
\overline{D}_T &= \frac{1}{m}\sum_{j=1}^J Z_j \overline{D}_j, \quad{} \overline{D}_C = \frac{1}{J-m}\sum_{j=1}^{J} (1 - Z_j)\overline{D}_j
\end{align*}
Here, $\overline{Y}_T$ and $\overline{Y}_C$ represent the average outcomes among treated and control clusters, respectively. Similarly, $\overline{D}_T$ and $\overline{D}_C$ represent the average treatment receipts among treated and control clusters. The four quantities, $\overline{Y}_T, \overline{Y}_C,\overline{D}_T$ and $\overline{D}_C$, are used to define the original Wald-like estimator \citep{wald1940} for $\tau$ in cluster settings, which we denote as $\widehat{\tau}_{{\rm CL}}$:
\begin{equation} \label{eq:cl}
\widehat{\tau}_{{\rm CL}} = \frac{\overline{Y}_T - \overline{Y}_C}{\overline{D}_T - \overline{D}_C}. 
\end{equation}
For testing the null hypothesis of $\tau$ in equation \eqref{eq:hyp}, \citet{schochet2013estimators} uses the Delta Method to derive asymptotic properties the estimator $\widehat{\tau}_{{\rm CL}}$. Specifically, consider the estimated variances  for $\overline{Y}_T - \overline{Y}_C$ and $\overline{D}_T - \overline{D}_C$ 
\[
\widehat{\Var}(\overline{Y}_T - \overline{Y}_C) = \frac{J S_Y^2}{m(J - m)},\quad{} \widehat{\Var}(\overline{D}_T - \overline{D}_C) = \frac{JS_D^2}{m(J-m)}
\]
\noindent where
\begin{align*}
S_Y^2 &= \frac{\sum_{j=1}^J Z_j(\overline{Y}_j - \overline{Y}_T)^2 + \sum_{j=1}^{J}(1 - Z_j) (\overline{Y}_j - \overline{Y}_C)^2}{J-2} \\
S_D^2 &= \frac{\sum_{j=1}^J Z_j (\overline{D}_j - \overline{D}_T)^2 + \sum_{j=1}^{J} (1 - Z_j) (\overline{D}_j - \overline{D}_C)^2}{J-2}
\end{align*}
Additionally, the estimated covariance between $\overline{Y}_T - \overline{Y}_C$ and $\overline{D}_T - \overline{D}_C$ is
\[
\widehat{\rm Cov}(\overline{Y}_T - \overline{Y}_C, \overline{D}_T - \overline{D}_C) =  \frac{\sum_{j=1}^{J} Z_j (\overline{Y}_j - \overline{Y}_T)(\overline{D}_j - \overline{D}_T)}{m^2} + \frac{\sum_{j=0}^{J} (1 - Z_j)(\overline{Y}_j - \overline{Y}_C)(\bar{D}_j - \overline{D}_C)}{(J -m)^2} 
\]
Based on the Delta method, \citet{schochet2013estimators} proposes the following estimator for the variance of $\hat{\tau}_{{\rm CL}}$ using the estimated variances and covariance
\[
\widehat{\Var}(\widehat{\tau}_{\rm CL})= \frac{ \Var(\overline{Y}_T - \overline{Y}_C)}{(\overline{D}_T - \overline{D}_C)^2} + \widehat{\tau}_{{\rm cl}}^2\frac{ \widehat{\Var}(\overline{D}_T - \overline{D}_C)}{(\overline{D}_T - \overline{D}_C)^2} - 2 \widehat{\tau}_{\rm cl} \frac{ \widehat{\rm Cov}(\overline{Y}_T - \overline{Y}_C, \overline{D}_T - \overline{D}_C)}{(\overline{D}_T - \overline{D}_C)^2}
\]
Also, for any $\alpha \in (0,1)$,  \citet{schochet2013estimators} suggested the following $1-\alpha$ confidence interval for $\tau$ 
\begin{equation} \label{eq:ci_cl}
\widehat{\tau}_{\rm CL} \pm z_{1-\alpha/2} \sqrt{\widehat{\Var}(\widehat{\tau}_{\rm CL})}
\end{equation}
where $z_{1-\alpha/2}$ is the $1-\alpha/2$ quantile of the standard Normal distribution. If $\tau_0$ in $H_0$ is included in interval \eqref{eq:ci_cl}, then $H_0$ is retained at the $\alpha$ level. Otherwise, $H_0$ is rejected in favor of the two-sided alternative $H_1: \tau \neq \tau_0$. 


\subsection{The Unit-Level Method} 
\label{sec:tsls}
Unit-level methods analyze CRTs with noncompliance at the unit level by directly using individual measurements rather than cluster level summary statistics. For inference, unit level methods use robust variance estimators that take into account intra-cluster correlations. This approach is a generalization of clustered standard errors developed by \citet{liang1986longitudinal}.

Formally, let $\mathbf{Y} = (Y_{11},Y_{12},\ldots,Y_{1n_1},Y_{21},\ldots,Y_{Jn_J})$ be the vectorized outcome variable, $\mathbf{D} = (D_{11},D_{12},\ldots,D_{1n_1},D_{21},\ldots,D_{Jn_J})$ be the vectorized compliance variable, and $\mathbf{Z} = (Z_{11},Z_{12},\ldots,Z_{1n_1},Z_{21},\ldots,Z_{Jn_J})$ be the vectorized treatment assignment. Point estimation of the CACE relies on two-stage least squares (TSLS) applied to these three vectors. The TSLS estimator $\widehat{\tau}_{\rm TSLS}$ has a closed-form expression and can be expressed as
\begin{align} \label{eq:TSLS}
\widehat{\tau}_{\rm TSLS} &= \frac{\frac{\sum_{j=1}^{J} Z_j Y_j}{\sum_{j=1}^{J} Z_j n_j} - \frac{\sum_{j=1}^{J} (1 - Z_j) Y_j}{\sum_{j=1}^J (1 - Z_j) n_j }}{\frac{\sum_{j=1}^{J} Z_j D_j}{\sum_{j=1}^{J} Z_j n_j} - \frac{\sum_{j=1}^{J} (1 - Z_j) D_j}{\sum_{j=1}^J (1 - Z_j) n_j }} \\
&=  \frac{ \sum_{j=1}^{J} (1-Z_j) n_j \sum_{j=1}^{J} Z_j Y_j - \sum_{j=1}^{J} Z_j n_j \sum_{j=1}^{J} (1 - Z_j)Y_j }{ \sum_{j=1}^{J} (1-Z_j) n_j \sum_{j=1}^{J} Z_j D_j - \sum_{j=1}^{J} Z_j n_j \sum_{j=1}^{J} (1 - Z_j)D_j }
\end{align}
For testing the null hypothesis of $\tau$ in equation \eqref{eq:hyp}, let $u_{ji} = Y_{ji} - D_{ji} \widehat{\tau}_{\rm TSLS}$ be the estimated residual, $u_{j} = \sum_{i=1}^{n_j} u_{ji}$ be the sum of residuals for each cluster, $\widehat{D}_{j} = \sum_{i=1}^{n_j} \widehat{D}_{ji}$ be the sum of predicted treatment receipts for each cluster, $\widehat{D}_{j}^2 = \sum_{i=1}^{n_j} \widehat{D}_{ji}^2$ be the sum of predicted treatment receipt square for each cluster, and $\widehat{D_j u_j} = \sum_{i=1}^{n_j } \widehat{D}_{ji} u_{ji}$ be the sum of the product of the residual and the predicted treatment receipt for each cluster. Then, following \citet{liang1986longitudinal}, \citet{wooldridge2010econometric} and \citet{cameron2015practitioner} and some algebra (see supplementary materials for details), the estimated cluster-robust standard error is
\[
\widehat{\Var}(\widehat{\tau}_{\rm TSLS}) =\frac{ \left(\sum_{j=1}^{J} \widehat{D}_j \right)^2 \left(\sum_{j=1}^{J} u_j^2 \right) + n^2 
\left(\sum_{j=1}^{J} \widehat{D_j u_j}^2 \right) - 2n \left(\sum_{j=1}^{J} \widehat{D}_{j} \right) \left( \sum_{j=1}^{J} u_j \widehat{D_j u_j} \right)  }{\left(n \sum_{j=1}^{J} \widehat{D}_{j}^2 - \left( \sum_{j=1}^{J} \widehat{D}_j \right)^2 \right)^2}
\]
Standard econometric arguments (see Chapter 5.2 of \citet{wooldridge2010econometric}) can be used to construct a $1-\alpha$ confidence interval for $\tau$
\begin{equation} \label{eq:ci_tsls}
\widehat{\tau}_{\rm TSLS} \pm z_{1-\alpha/2} \sqrt{\widehat{\Var}(\widehat{\tau}_{\rm TSLS}) }
\end{equation}
Prior simulation studies suggest that $40 \leq J $ is necessary for this confidence interval to be reasonable \citep{Angrist:2009}. See recent work by \citet{pustejovskycluster2018} for improved small sample performance.


\section{Properties of Extant Methods}
\subsection{Fixed Population} \label{sec:fixed_pop}
We study the properties of the methods in Sections \ref{sec:clus} and \ref{sec:tsls}, specifically $\widehat{\tau}_{\rm CL}$ and $\widehat{\tau}_{\rm TSLS}$, when the sample size $n$ and population $\mathcal{F}$ are fixed. Formally, we follow \citet{Angrist:1996}, \citet{schochet2013estimators}, \citet{imbens_causal_2015} and study the first-order (Taylor) properties of these two estimators when the population remains fixed. Specifically, consider the first-order Taylor approximation of the function $f(x,y) = x/y$ centered at values $x_0$ and $y_0$
\begin{align*}
f(x,y) &= f(x_0,y_0) + \frac{\delta}{\delta x} f(x_0,y_0) (x - x_0) + \frac{\delta}{\delta y} f(x_0,y_0)(y - y_0) + R \\
         &\approx f(x_0,y_0) + \frac{\delta}{\delta x} f(x_0,y_0) (x - x_0) + \frac{\delta}{\delta y} f(x_0,y_0)(y - y_0)\\
         &= \frac{x_0}{y_0} + \frac{1}{y_0}(x - x_0) - \frac{x_0}{y_0^2} (y - y_0)
\end{align*}
where $R$ is the remainder term from the Taylor approximation. If $x$ and $y$ are random variables and we set $x_0$ and $y_0$ to be the expectations of the numerator and denominator, respectively, we can use the approximation to characterize the central tendency of both estimators. We remark that the approximation sign does not indicate the approximation coming from the sample size going to infinity, but rather indicates the approximation coming from the truncation of the Taylor series. Also, while not as precise as higher-order approximations, we find that the first-order terms provide sufficient insights into the finite-sample properties of the two estimators. 

Proposition \ref{lem:cl} shows that the first-order properties of the two methods depend on either the assumption that (i) the number of units in all clusters are the same or (ii) cluster-level CACEs are identical across all clusters to estimate $\tau$.
\begin{proposition} \label{lem:cl} If Assumptions \ref{a1}-\ref{a4} hold, the first-order expectation of $\widehat{\tau}_{\rm CL}$ and $\widehat{\tau}_{\rm TSLS}$, denoted as $\tau_{\rm CL}$ and $\tau_{\rm TSLS}$, respectively, are weighted complier average treatment effects across each cluster-specific CACE $\tau_j$.
\begin{align*}
\tau_{\rm CL} &= \sum_{j=1}^{J} w_{{\rm CL},j} \cdot \tau_j, \quad{} w_{{\rm CL}, j}= \frac{\frac{n_{{\rm CO},j}}{n_j}}{\sum_{j=1}^{J} \frac{n_{{\rm CO},j}}{n_j}}\\
\tau_{\rm TSLS} & = \sum_{j=1}^{J}  w_{{\rm TSLS},j} \tau_j, \quad{} w_{{\rm TSLS},j} = \frac{n_{{\rm CO},j} (n - n_j)}{\sum_{j=1}^{J} n_{{\rm CO},j} (n-n_j)} 
\end{align*}
\end{proposition}
Proposition \ref{lem:cl} demonstrates how each method of estimation weighs cluster-level complier average treatment effects $\tau_j$. While $\tau_{\rm CL}$ averages each cluster-level CACE $\tau_j$ by the proportion of compliers, $n_{{\rm CO},j}/n_j$, it does not take into consideration the size of each cluster. Alternatively, $\tau_{\rm TSLS}$ averages each cluster-level CACE $\tau_j$ by the product of the number of compliers per cluster $n_{{\rm CO},j}$ and the number of individuals in other clusters $n - n_j$, via $n_{{\rm CO},j} (n - n_j )$. Both weights are different than the true CACE $\tau$, which weights by the number of compliers $n_{{\rm CO},j}$ via $n_{{\rm CO},j} / \sum_{j=1}^J n_{{\rm CO},j}$. 

Tables \ref{tab:CACE1} and \ref{tab:CACE2} provide a simple numerical demonstration of how the weight of each estimator compare to the true $\tau$ in a CRT with $J = 3$ clusters. The three clusters have $n_1 = 80, n_2= 10$, and $n_3 = 10$ units and the cluster level complier average treatment effects are $\tau_1 = 1$, $\tau_2 = 2$, and $\tau_3 = 1.5$.  In Table \ref{tab:CACE1}, the compliance rates across the three clusters are identical and set to 50\%. Under this scenario, the true CACE $\tau$ in equation \eqref{eq:CATE} weighs the cluster-level CACE proportional to the number of compliers in each cluster with weights $w_{1} = 0.8$ and $w_{2} = w_{3} = 0.1$. However, the cluster-based method $\tau_{\rm CL}$ gives disproportionately large weights to small clusters because the compliance rates are identical and all clusters' CACEs, despite their differences in size, are weighted equally with weights $w_{{\rm CL},1} = w_{{\rm CL},2} = w_{{\rm CL},3} = 1/3$. The unit-based method $\tau_{\rm TSLS}$ also gives disproportionately large weights to the small clusters. However, unlike $\tau_{\rm CL}$, it takes the size of each cluster into consideration when weighting the cluster-level CACEs and clusters with different sizes get slightly different weights: $w_{{\rm TSLS}, 1} \approx 800/1700$ and $w_{{\rm TSLS},2} = w_{{\rm TSLS},3} \approx 450/1700$. Also, the two clusters of identical size receive identical weights under the unit-level method, similar to the weights from the true CACE; in contrast, the cluster-level method ignores the size of the cluster and weights each cluster equally. Nevertheless, $\tau_{\rm CL}$ and $\tau_{\rm TSLS}$ under-estimate the true CACE in equation \eqref{eq:CATE}. 

\begin{table}[h!] 
\begin{center}
\begin{threeparttable}
\caption{A numerical example of a CRT with three clusters $J=3$ and identical compliance rates.} 
\label{tab:CACE1}
\begin{tabular}{p{1.2cm} | p{1.5cm} p{4cm} p{8.5cm}}

Cluster & Cluster Size: $n_j$ & Number of Compliers: $n_{{\rm CO},j}$ (\% compl. rate)& CACE per Cluster: $\tau_j$ \\ \toprule
$j= 1$ & $n_1 = 80$ & $n_{{\rm CO},1} = 40$ (50\%) & $\tau_{1} = 1$ \\
$j = 2$ & $n_2 = 10$ & $n_{{\rm CO},2} = 5$  (50\%) & $\tau_{2} = 2$ \\ 
$j = 3$ & $n_3 = 10$ &  $n_{{\rm CO},3} = 5$ (50\%) & $\tau_{3} = 1.5$ \\ \midrule
 $J = 3$ &$n = 100$ & $n_{\rm CO} = 50$ & $\tau = 1 (\frac{40}{50}) + 2(\frac{5}{50}) + 1.5 (\frac{5}{50}) = 1.15$ \\
\addlinespace[1ex] && & $\tau_{\rm CL} = 1 (\frac{0.5}{1.5}) + 2(\frac{0.5}{1.5}) + 1.5(\frac{0.5}{1.5}) =  1.5$\\ 
\addlinespace[1ex] && & $\tau_{\rm TSLS} = 1 (\frac{800}{1700}) + 2(\frac{450}{1700}) + 1.5 (\frac{450}{1700}) \approx 1.4$\\
\bottomrule
\end{tabular}
\begin{tablenotes}[para]
\small{Note: $\tau$ is the true CACE, $\tau_{\rm CL}$ is the identified value based on the cluster-level estimator $\widehat{\tau}_{\rm CL}$, and $\tau_{\rm TSLS}$ is the identified value based on TSLS $\widehat{\tau}_{\rm TSLS}$.}
\end{tablenotes}
\end{threeparttable}
\end{center}
\end{table}
 In Table \ref{tab:CACE2}, the setting is identical to Table \ref{tab:CACE1} except the compliance rates vary across the three clusters. Specifically, the two small clusters have higher compliance rates compared to the large cluster (80\% versus 10\%), but the number of compliers remain identical across clusters. Similar to Table \ref{tab:CACE1}, we see that both $\tau_{\rm CL}$ and $\tau_{\rm TSLS}$ tend to up-weigh the two small clusters and down-weigh the large cluster in comparison to the true $\tau$. We also see that in both tables, $\tau_{\rm TSLS}$ tends to be closer to the true CACE $\tau$ than $\tau_{\rm CL}$; however both methods tend to overestimate the true CACE.

\begin{table}[h!] 
\begin{center}
\begin{threeparttable}
\caption{A numerical example of a CRT with three clusters $J=3$ and different compliance rates.} 
\label{tab:CACE2}

\begin{tabular}{p{1.2cm} | p{1.5cm} p{4cm} p{8.5cm}}
Cluster & Cluster Size: $n_j$ & Number of Compliers: $n_{{\rm CO},j}$ (\% compl. rate)& CACE per Cluster: $\tau_j$ \\ \toprule
$j= 1$ & $n_1 = 80$ & $n_{{\rm CO},1} = 8$ (10\%) & $\tau_{1} = 1$ \\
$j = 2$ & $n_2 = 10$ & $n_{{\rm CO},2} = 8$  (80\%) & $\tau_{2} = 2$ \\ 
$j = 3$ & $n_3 = 10$ &  $n_{{\rm CO},3} = 8$ (80\%) & $\tau_{3} = 1.5$ \\ \midrule 
$J = 3$ &$n = 100$ & $n_{\rm CO} = 24$ & $\tau = 1 (\frac{8}{24}) + 2(\frac{8}{24}) + 1.5(\frac{8}{24}) = 1.5 $\\
\addlinespace[1ex] && & $\tau_{\rm CL} = 1 (\frac{0.1}{1.7}) + 2(\frac{0.8}{1.7}) + 1.5(\frac{0.8}{1.7}) \approx  1.71 $\\
\addlinespace[1ex] && & $\tau_{\rm TSLS} = 1 (\frac{320}{1600}) + 2(\frac{720}{1760}) + 1.5 (\frac{720}{1760}) \approx 1.68$\\
\bottomrule
\end{tabular}
\begin{tablenotes}[para]
\small{Note: $\tau$ is the true CACE, $\tau_{\rm CL}$ is the identified value based on the cluster-level estimator $\widehat{\tau}_{\rm CL}$, and $\tau_{\rm TSLS}$ is the identified value based on TSLS $\widehat{\tau}_{\rm TSLS}$.}
\end{tablenotes}
\end{threeparttable}
\end{center}
\end{table}

Corollary \ref{lem:equiv} shows that ${\tau}_{\rm CL}$ and ${\tau}_{\rm TSLS}$ are equal to the true CACE $\tau$ under two different conditions: (1) if the clusters are of equal size or (2) if the cluster-level CACEs are identical across clusters.
\begin{corollary} \label{lem:equiv} Suppose either condition below holds:
\begin{enumerate}
\item The number of units in each cluster $n_j$ is identical across all clusters.
\item The complier average treatment effect for each cluster $\tau_j$ is homogeneous/identical across all clusters.
\end{enumerate}
Then, $\tau_{\rm CL}$ and $\tau_{\rm TSLS}$ equal the complier average treatment effect $\tau$. 
\end{corollary}
The first condition in Corollary \ref{lem:equiv} allows for heterogeneity in complier average treatments across clusters, but stipulates that the cluster size must be the same for $\tau_{\rm CL}$ and $\tau_{\rm TSLS}$ to equal $\tau$. The second condition allows the cluster size to vary, but stipulates that every cluster must have the same cluster-level CACE for $\tau_{\rm CL}$ and $\tau_{\rm TSLS}$ to equal $\tau$. The latter condition is common in traditional econometrics by assuming a linear structural equation for the $Y$ and $D$ relationship where the coefficient associated with $D$ and $Y$ are assumed to be constant \citep{wooldridge2010econometric}. Neither is realistic in my CRT applications, including the smartcard intervention of interest in our study.

In summary, under a fixed population analysis, the two popular methods for analyzing CRT data with noncompliance may not always evaluate the target causal parameter of interest, CACE ($\tau$), even under a first-order analysis. The weights that make up both ${\tau}_{\rm CL}$ and ${\tau}_{\rm TSLS}$, $w_{{\rm CL,j}}$ and $w_{{\rm TSLS},j}$ respectively, which differ from the weights of $\tau$ in equation \eqref{eq:tau_j}. Indeed, only under the restrictive conditions stated in Corollary \ref{lem:equiv} can the two popular methods actually evaluate the CACE.

\subsection{Growing Population}
In the previous section, we studied the properties of the two methods assuming the population $\mathcal{F}$ is fixed. However, for hypothesis testing, often one assumes the sample size grows to infinity and consequently, $\mathcal{F}$ is growing. Here, we consider whether $\widehat{\tau}_{\rm CL}$ and $\widehat{\tau}_{\rm TSLS}$ estimate the true CACE $\tau$ under the large sample regime.

In a CRT, there are primarily two ways in which we might imagine the sample size growing. The first way is where the number of clusters $J$ is getting larger, but the units per cluster $n_{j}$ remain bounded. This regime is the basis for justifying the testing properties of many estimators in the CRT literature, including both $\widehat{\tau}_{\rm CL}$ and $\widehat{\tau}_{\rm TSLS}$ as mentioned before. This regime also serves as the basis for justifying the inferential properties of many causal estimators outside of the CRT literature; see \citet{li_general_2017} for a recent review. In practice,  this asymptotic regime is a good approximation of CRTs in political science concerning effectiveness of voter mobilization strategies where households serve as clusters and the number of households is large compared to the number of individuals that make up the household; see \citet{Gerber:2008} for an example and \citet{Green:2013} for an overview.

The second way in which the sample size can grow is  where the number of units per clusters $n_{j}$ is growing, but the number of clusters $J$ remains bounded. Under this regime, the inferential properties (e.g. confidence intervals) for many CRT estimators, including both $\widehat{\tau}_{\rm CL}$ and $\widehat{\tau}_{\rm TSLS}$ are difficult to characterize because the standard theoretical device based on the Lindeberg central limit theorem is no longer applicable. In practice, some CRT designs are closer to this regime. For example, \citet{hayes2014hptn} outlines a CRT design with 21 clusters with approximately 55,000 units in each cluster. 
 \citet{solomon2015high} describe a design with 12 cities as clusters and 1000 individuals per cluster. Generally speaking, a design of this type is more common in clustered observational studies; see \citet{acemoglu2000large} for one example. 

Let $p_{{\rm CO},j} = n_{\rm CO,j} / n_j$ be the proportion of compliers in cluster $j$. Proposition \ref{prop:asym_J} examines the properties of the two methods under the first asymptotic regime where the number of clusters go to infinity and the number of units per cluster remain bounded. Under this asymptotic regime, only $\tau_{\rm TSLS}$ is consistent for CACE whereas $\tau_{\rm CL}$ may not be.
\begin{proposition} \label{prop:asym_J} Suppose Assumptions \ref{a1}-\ref{a4} hold. As $n$ grows, consider a sequence of $\mathcal{F}$ where (i) $J\to \infty$, (ii) $n_j$s are bounded, (iii) the average cluster-level potential outcomes $\overline{Y}_{j}^{\left(z,D_{j}^{(z)}\right)}$ are bounded, and (iv) the compliance rate $\mu_D$ is bounded away from zero, i.e. $\mu_D > \delta$ for some constant $\delta > 0$.  Then, $\widehat{\tau}_{\rm CL}$ and the first-order expectation of $\widehat{\tau}_{\rm CL}$ is asymptotically biased away from $\tau$ by $b_n$, i.e
\begin{align*}
\lim_{n \to \infty} \widehat{\tau}_{\rm CL} - \tau -b_n  = 0, \quad{} \lim_{n \to \infty} \tau_{\rm CL} - \tau - b_n = 0
\end{align*}
where
\[
 b_n = \frac{1}{\left(\sum_{j=1}^J p_{ {\rm CO},j} \right) \left( \sum_{j=1}^J n_{{\rm CO},j} \right) }\sum_{j < l} p_{\rm CO,j} p_{\rm CO,l} (n_l - n_j) (\tau_j - \tau_l)
\]
But, $\widehat{\tau}_{\rm TSLS}$ and the first-order expectation of $\widehat{\tau}_{\rm TSLS}$ converges to $\tau$:
\[
\lim_{n \to \infty} \widehat{\tau}_{\rm TSLS} - \tau =  0, \quad{}  \lim_{n \to \infty} \tau_{\rm TSLS} - \tau  =  0
\]
\end{proposition}
Unlike the results in Propositions \ref{lem:cl} which were based on finite samples and where the two methods may not always estimate $\tau$, Proposition \ref{prop:asym_J} states that $\widehat{\tau}_{\rm TSLS}$ converges to the limiting $\tau$ under the aforementioned asymptotic regime. However, $\widehat{\tau}_{\rm CL}$ may not converge to the limiting $\tau$. For example, suppose in a household survey data, each household defines a cluster and the total number of people per household, $n_j$, is bounded by $B = 4$. Let half of the households have two individuals, $n_j = 2$, while the other half have four individuals $n_j = 4$; both types of households have 50\% compliance rate, $p_{{\rm CO},} = 0.5$. Also, for the two-member households, their CACEs are $\tau_{j} = 4$ while for the four-member household, their CACEs are $\tau_{j} = 2$. Then, the asymptotic bias of $\widehat{\tau}_{\rm CL}$ is
{\small
\begin{align*}
\frac{1}{\left(\sum_{j=1}^J p_{ {\rm CO},j} \right) \left( \sum_{j=1}^J n_{{\rm CO},j} \right) }\sum_{j < l} p_{\rm CO,j} p_{\rm CO,l} (n_l - n_j) (\tau_j - \tau_l) =& \frac{1}{\left(\frac{J}{2} \right) \left(\frac{J}{2} 2 + \frac{J}{2} 4 \right)} \sum_{j < l, n_l \neq n_j} \frac{1}{4} (2)(2) =\frac{1}{6}
\end{align*}}
We can increase the asymptotic bias arbitrary large by selecting different values of $\tau_j$. But, $\widehat{\tau}_{\rm CL}$ will converge to $\tau$ if conditions in Corollary \ref{lem:equiv} are satisfied. Combined together, $\widehat{\tau}_{\rm CL}$ is sensitive to the design as well as the underlying (unobserved) heterogeneity of CACE across clusters whereas $\widehat{\tau}_{\rm TSLS}$ is asymptotically insensitive to the cluster design.

Proposition \ref{prop:asym_J} also has broader implications. First, the result in Proposition \ref{prop:asym_J} underscores the role of asymptotic assumptions for identifying the target causal parameter $\tau$ of interest, especially for $\widehat{\tau}_{\rm TSLS}$. Second, because there may be CRTs for which $\widehat{\tau}_{\rm CL}$ will never consistently estimate the CACE, this suggest that doing asymptotic inference using clustered-based method, say computing p-values or confidence intervals in Section \ref{sec:clus}, may be invalid (i.e. inflated Type I errors or incorrect coverage).

Next, Proposition \ref{prop:asymp_nj} examines the properties of $\widehat{\tau}_{\rm CL}$ and $\widehat{\tau}_{\rm TSLS}$ in the asymptotic regime of the second type where the number of clusters $J$ remains bounded and the number of units per cluster $n_j$ goes to infinity. Unlike the regime in Proposition \ref{prop:asym_J}, the limiting variables for $\widehat{\tau}_{\rm CL}$ and $\widehat{\tau}_{\rm TSLS}$ when $n_j \to \infty$ will not be constant, but random variables.  As such, we study the limiting behavior of the expectations of the first-order Taylor approximations of $\widehat{\tau}_{\rm CL}$ and $\widehat{\tau}_{\rm TSLS}$, which provides an approximation about the moment of the limiting values of $\widehat{\tau}_{\rm CL}$ and $\widehat{\tau}_{\rm TSLS}$.
\begin{proposition} \label{prop:asymp_nj} Suppose Assumptions \ref{a1}-\ref{a4} hold. As $n$ grows, consider a sequence of $\mathcal{F}$ where (i) $J$ remains bounded, (ii) for every $j$, $n_j, n_{{\rm CO},j} \to \infty$ where $n_{{\rm CO},j} / n_j \to \tilde{p}_{{\rm CO},j} \in (0,1)$ and $n_j / n \to \tilde{p}_{j} \in (0,1)$, and (iii) $\tau_{j} \to \tilde{\tau}_{j}$. Then, the first-order expectations of the cluster-based and unit-based estimator, $\tau_{\rm CL}$ and $\tau_{\rm TSLS}$, converge to the following values 
\[
\lim_{n \to \infty} \tau_{\rm CL} - \tau = \sum_{j=1}^{J} \tilde{\tau}_{j}  \tilde{p}_{{\rm CO},j} \left( \frac{ \tilde{p}_{j} }{\sum_{l=1}^{J} \tilde{p}_{{\rm CO},l} \tilde{p}_{l}}-  \frac{ 1}{\sum_{l=1}^{J} \tilde{p}_{{\rm CO},l}}\right) 
\]
and 
\[
\lim_{n \to \infty} \tau_{\rm TSLS} - \tau = \sum_{j=1}^{J} \tilde{\tau}_{j}  \tilde{p}_{{\rm CO},j} \left( \frac{ \tilde{p}_{j} (1 - \tilde{p}_{j})}{\sum_{l=1}^{J} \tilde{p}_{{\rm CO},l} \tilde{p}_{l} (1 - \tilde{p}_l) }-  \frac{ 1}{\sum_{l=1}^{J} \tilde{p}_{{\rm CO},l}}\right)
\]
\end{proposition}
Proposition \ref{prop:asymp_nj} states that under the growing $n_{j}$ asymptotic regime, both $\tau_{\rm CL}$ and $\tau_{\rm TSLS}$ may not converge to $\tau$. But, both behave differently in terms of how they diverge away from $\tau$. As an illustration, if we have $J =4 $ clusters where the proportions of compliers $p_{{\rm CO},j}$ across clusters are identical and the limiting cluster-level CACEs $\tau_{j}$ are non-negative, we arrive at 
\[
\lim_{n \to \infty} \tau_{\rm CL} - \tau = \sum_{j=1}^{4} \tilde{\tau}_j \left(\tilde{p}_{j} - \frac{1}{4} \right), \quad{} \lim_{n \to \infty} \tau_{\rm TSLS} - \tau =  \sum_{j=1}^{4} \tilde{\tau}_j \left( \frac{\tilde{p}_{j}(1-\tilde{p}_{j})}{ \sum_{l=1}^{4} \tilde{p}_{l}(1-\tilde{p}_{l})}  - \frac{1}{4} \right)
\]
For the cluster-level method, large clusters whose proportion $\tilde{p}_{j}$ exceed $1/4$, which is the proportion when all four clusters have equal size, bias the estimate upwards. In contrast, small clusters whose proportion $\tilde{p}_{j}$ is under $1/4$ bias the estimate downwards; in fact, large clusters tend to contribute more to the bias than small clusters because the positive maximal deviation of $\tilde{p}_{j}$ from $1/4$ is larger than the negative maximal deviation of it. For the unit-level method, the direction of bias is similar to the cluster-level method, however the magnitude of the bias is generally smaller. For example, if $\tilde{p}_{1} = \tilde{p}_{2} = 0.1$ and $\tilde{p}_3 = 0.3$ and $\tilde{p}_4 = 0.5$, the maximal positive and negative differences of $\tilde{p}_{j} - 1/4$ in the cluster-level estimator are $-0.15$ and $0.25$, respectively. In contrast, the maximal positive and negative differences of $\tilde{p}_{j}(1-\tilde{p}_{j}) / \sum_{l=1}^{4} \tilde{p}_{l}(1-\tilde{p}_{l})  - \frac{1}{4}$ in the unit-level estimator are $-0.11$ and $0.14$, respectively. This is because of the quadratic term in the unit-level estimator $\tilde{p}_{j}(1-\tilde{p}_{j})$, which essentially shrinks the bias towards a common mean. Note that following Corollary \ref{lem:equiv}, if the cluster sizes are asymptotically identical so that $\tilde{p}_{j} = \tilde{p}_{l}$ for all $j\neq l$, the two methods will converge to the limiting values of $\tau$. 

In summary, in asymptotic settings where the number of clusters is growing and the number of units within each cluster are bounded, the finite-sample identification problems of $\widehat{\tau}_{\rm TSLS}$ go away. However, $\widehat{\tau}_{\rm CL}$ may still fail to converge to the true CACE $\tau$. In contrast, in the asymptotic setting where the number of units within each cluster grow and the number of clusters remain fixed, the finite sample problems discussed in Section \ref{sec:fixed_pop} remain, and the two methods may not converge to the true CACE.

We conclude the section by making a few technical points. First, condition (iii) in Proposition \ref{prop:asym_J} can be relaxed so that the average cluster-level potential outcome is growing, but at a rate slower than $J$. We do not believe this scenario is realistic in practice for the bounded cluster-size asymptotic regime and hence, we state a more simple boundedness condition. Second, in the proof of Proposition \ref{prop:asym_J}, we show the rate of convergence of $\widehat{\tau}_{\rm TSLS}$ is $\sqrt{J}$, which is the usual rate for this asymptotic regime. Finally, for clarity of exposition, we avoid a more technically accurate exposition where we place subscript $n$ in $\mathcal{F}$, $\tau_{\rm CL}$, $\tau_{\rm TSLS}$, and $\tau$.

\subsection{Inference}

Next, we review the inferential properties of  $\widehat{\tau}_{\rm CL}$ and $\widehat{\tau}_{\rm TSLS}$, which have been more thoroughly discussed in the literature than our discussion on identification and point estimation above \citep{schochet2013estimators,wooldridge2010econometric}; often, the literature simply assumes that the proposed estimators estimate the quantity of interest $\tau$ and proceed with inference. Suppose that $\widehat{\tau}_{cl}$ and $\widehat{\tau}_{tsls}$ identify the CACE $\tau$ under additional assumptions in Corollary \ref{lem:equiv}. Both rely on asymptotic approximations in the asymptotic regime specified in Proposition \ref{prop:asym_J} for inference where the number of clusters go to infinity, the number of units within each cluster remain bounded, and the treatment assignment has a strong effect on compliance; the latter condition is equivalent to having experiments with uniformly high compliance rates. If the experiment has low compliance rates, then the asymptotic arguments used to construct these confidence intervals are no longer accurate and will often be shorter and/or off from the true target value; see \citet{stock2002survey} and references therein for examples. In general, we would argue that the conditions of the smartcard intervention do not closely hew to any of the asymptotic templates. The number of clusters is not particularly large, the number of units vary from cluster to cluster, and the compliance rates vary from cluster to cluster.

\section{An Almost Exact Approach}

Thus far we have outlined two particular weaknesses in extant methods. First, these methods do not always estimate CACE, especially in finite samples, and for one approach, asymptotics assumptions have to be used if identification is desired. This is a particular concern, since some CRTs have small sample sizes. Second, both methods rely on inferential techniques that require the compliance rate to be high. Next, we propose solutions that estimates CACE even if cluster sizes or cluster-specific CACE differ and can produce confidence intervals that remain valid even if the compliance rate is low. Our solutions are adapted from the almost exact methods of \citet{Keele:2017fiv} for non-clustered experimental designs, which uses closed-form expressions for interval estimation and is a specific instance of ``finite sample asymptotics'' laid out in \citet{Hajek:1960} and \citet{Lehmann:2004}. 

Formally, for any $\tau_0 \in \reals$ and for each cluster $j$, let $A_j(\tau_0) = Y_j - D_j \tau_0$ be the adjusted response for each cluster, $A_j^{(1)}(\tau_0) = Y_{j}^{\left(1,D_j^{(1)}\right)} - D_{j}^{(1)} \tau_0$ be the adjusted potential response if cluster $j$ was treated, and $A_j^{(0)}(\tau_0) = Y_{j}^{\left(0,D_j^{(0)}\right)} - D_{j}^{(0)} \tau_0$ be the adjusted potential response if cluster $j$ was not treated. Let $A_T(\tau_0) = \frac{1}{m} \sum_{j=1}^{J} Z_j (Y_j - D_j \tau_0)$ be the mean of the adjusted response among the treated and let $A_C(\tau_0) = \frac{1}{J - m} \sum_{j=1}^{J} (1 - Z_j) (Y_j - D_j \tau_0)$ be the mean of the adjusted response among the control. Then, we propose the test statistic $T(\tau_0)$ for the hypothesis in equation \eqref{eq:hyp}
\begin{align*}
T(\tau_0) &= \frac{1}{m} \sum_{j=1}^{J} Z_{j} (Y_j - D_{j} \tau_0) - \frac{1}{J- m} \sum_{j=1}^J (1 - Z_j) (Y_j - D_j \tau_0) \\
&= \frac{1}{m} \sum_{j=1}^{J} Z_j A_j (\tau_0) - \frac{1}{J - m} \sum_{j=1}^{J} (1 - Z_j) A_j(\tau_0) \\
&= A_{T}(\tau_0) - A_{C}(\tau_0)
\end{align*}
Proposition \ref{prop:ours} shows that the test statistic $T(\tau_0)$ under appropriate scaling is asymptotically Normal.
\begin{proposition} \label{prop:ours} Suppose Assumptions (A1)-(A4) and the null hypothesis $H_0: \tau  = \tau_0$ holds. Then, as $J, J - m \to \infty$ where $m / J \to p \in (0,1)$, suppose (i) $\mu_Y$ and $\mu_D$ are constants so that $\mu_Y/\mu_D = \tau_0$ and (ii) the following growth condition holds
\[
\frac{\max_{j} \left( \frac{A_j^{(1)}}{m} + \frac{A_j^{(0)}}{J - m} -  \left( \frac{1}{J} \sum_{j=1}^{J}  \frac{A_j^{(1)}}{m} + \frac{A_j^{(0)}}{J - m} \right) \right)^2}{\sum_{j=1}^{J}  \left( \frac{A_j^{(1)}}{m} + \frac{A_j^{(0)}}{J - m} -  \left( \frac{1}{J} \sum_{j=1}^{J}  \frac{A_j^{(1)}}{m} + \frac{A_j^{(0)}}{J - m} \right) \right)^2} \to 0
\]
Then, we have
\[
\frac{T(\tau_0)}{\sqrt{Var(T(\tau_0) \mid \mathcal{F},\mathcal{Z})}} \to N(0,1)
\]
\end{proposition}
Unlike the growing $J$ asymptotic regime in Proposition \ref{prop:asym_J}, Proposition \ref{prop:ours} does not make the assumption that the compliance rate $\mu_D$ is far away from zero. Also, a consequence of Proposition \ref{prop:ours} is that we can construct a Hodges-Lehmann type of point estimator of $\tau$ based on the testing properties of $T(\tau_0)$  \citep{hodges_estimates_1963}. Specifically, our point estimator, which we denote as $\widehat{\tau}_{rm AE}$, is the value of $\tau$ that satisfies the equation $T(\tau) = 0$, which turns out to be
\begin{equation} \label{eq:tau_AE}
\widehat{\tau}_{\rm AE} = \frac{\frac{1}{m} \sum_{j=1}^{J} Z_j Y_j - \frac{1}{J-m} \sum_{j=1}^{J} (1 -Z_j) Y_j}{\frac{1}{m} \sum_{j=1}^{J} Z_j D_j - \frac{1}{J-m} \sum_{j=1}^{J} (1 -Z_j) D_j}
\end{equation}
The following corollary shows that the expectation of the first-order Taylor expansion of the proposed point estimator $\widehat{\tau}_{\rm AE}$ is the CACE in finite samples, irrespective of the number of units within the cluster or cluster-level CACE heterogeneity; this is in contrast to the two extant methods. Note that one can also derive the same estimator using the methods in \citet{middleton_unbiased_2015} as plug-in estimates for the terms in equation~\ref{eq:tau_AE}.

\begin{corollary} \label{coro:AE_ident} If Assumptions \ref{a1}-\ref{a4} hold, the first-order expectation of $\widehat{\tau}_{\rm AE}$ in equation \eqref{eq:tau_AE} is the complier average treatment effect, i.e.
\begin{align*}
\tau_{\rm AE} = \frac{\E[\frac{1}{m} \sum_{j=1}^{J} Z_j Y_j - \frac{1}{J-m} \sum_{j=1}^{J} (1 -Z_j) Y_j \mid \mathcal{F}, \mathcal{Z}] }{\E[ \frac{1}{m} \sum_{j=1}^{J} Z_j D_j - \frac{1}{J-m} \sum_{j=1}^{J} (1 -Z_j) D_j \mid \mathcal{F}, \mathcal{Z}]} = \tau 
\end{align*}
Also, under the asymptotic regime in Proposition \ref{prop:asym_J}, the estimator is consistent for $\tau$, i.e. $\lim_{n \to \infty} \widehat{\tau}_{\rm AE} - \tau = 0$. Also, under the asymptotic regime in Proposition \ref{prop:asymp_nj}, we have $\lim_{n \to \infty} \tau_{\rm AE} - \tau = 0$. 
\end{corollary}

Another consequence of Proposition \ref{prop:ours} is that for any $\alpha, 0 < \alpha <1$, we can construct a two-sided $1-\alpha$ confidence interval for $\tau$ by inverting the test and using the asymptotic Normal distribution under the null hypothesis. This requires specifying an estimator for $Var(T(\tau_0) \mid \mathcal{F},\mathcal{Z})$ and one such estimator is the classic sum of variance between the treated and control units \citep{Neyman:1923a,imbens_causal_2015}:
\begin{align*}
S^2(\tau_0) =& \frac{1}{m(m-1)} \sum_{j=1}^{J} Z_j (A_j (\tau_0) - A_T(\tau_0))^2 + \frac{1}{(J-m)(J-m-1)} \sum_{j=1}^{J} (1 - Z_j) (A_j(\tau_0) -  A_C(\tau_0))^2.
\end{align*}
The variance estimator $S^2(\tau_0)$ is well-known to be conservative and generally speaking, there does not exist an unbiased estimator for $Var[T(\tau_0) \mid \mathcal{F},\mathcal{Z}]$  \citep{Neyman:1923a}. Recent proposals by \citet{robins_confidence_1988} and \citet{aronow_sharp_2014} provide sharper estimates of $Var[T(\tau_0) \mid \mathcal{F},\mathcal{Z}]$, which can also be used in our context. For example, we can replace our variance estimate $S^2(\tau_0)$ by the upper bound $\widehat{V}_{N}^H$ in equation (9) of \citet{aronow_sharp_2014}. Regardless, once we have selected an estimator for variance, say $S^2(\tau_0)$, a two-sided $1-\alpha$ confidence interval of $\tau$ is the set of $\tau_0$ that is accepted under the null hypothesis.
\begin{equation} \label{eq:ci_broad}
\left\{\tau_0 : P_{H_0}\left(\left| \frac{T(\tau_0)}{S(\tau_0)}\right| \leq z_{1-\alpha/2} | \mathcal{F}, \mathcal{Z}\right)  \right\}
\end{equation}
where $z_{1-\alpha/2}$ is the $1-\alpha/2$ quantile of the standard Normal distribution. 
The interval produces valid statistical inference in that the interval in equation \eqref{eq:ci_broad} will cover $\tau$ with at least $1-\alpha$ probability. Also, the $1-\alpha$ confidence interval in equation \eqref{eq:ci_broad} can be greatly simplified as a solution to a quadratic equation. Consider the following quantities.
\begin{align*} 
\widehat{s}_{Y_T}^2 &= \frac{1}{m-1} \sum_{j=1}^{J} Z_j \left(Y_j - \frac{\sum_{j=1}^{J} Z_j Y_j}{m}\right)^2 \\
\widehat{s}_{Y_C}^2 &= \frac{1}{J - m-1} \sum_{j=1}^{J} (1 - Z_j) \left(Y_j - \frac{\sum_{j=1}^{J} (1 - Z_j) Y_j}{J-m}\right)^2\\
\widehat{s}_{D_T}^2 &= \frac{1}{m-1} \sum_{j=1}^{J} Z_j \left(D_j - \frac{\sum_{j=1}^{J} Z_j Y_j}{m}\right)^2 \\
\widehat{s}_{D_C}^2 &= \frac{1}{J - m-1} \sum_{j=1}^{J} (1 - Z_j) \left(D_j - \frac{\sum_{j=1}^{J} (1 - Z_j) D_j}{J-m}\right)^2 \\
\widehat{s}_{{YD}_T}^2 &= \frac{1}{m - 1} \sum_{j=1}^{J} Z_j \left(Y_j - \frac{\sum_{j=1}^{J} Z_j Y_j}{m} \right)\left(D_j - \frac{\sum_{j=1}^{J}Z_jD_j}{m} \right) \\
\widehat{s}_{{YD}_C}^2 &=\frac{1}{J - m - 1} \sum_{j=1}^{J} (1 - Z_j) \left(Y_j - \frac{\sum_{j=1}^{J} (1 - Z_j)Y_j}{J-m} \right)\left(D_j - \frac{\sum_{j=1}^{J} (1-Z_j)D_j}{J-m} \right)
\end{align*}
The terms $\widehat{s}_{Y_{T}}^2$ and $\widehat{s}_{D_T}^2$ are the estimated variances of $Y$ and $D$, respectively, among treated clusters. The terms $\widehat{s}_{Y_{C}}^2$ and $\widehat{s}_{D_C}^2$ are the estimated variances of $Y$ and $D$, respectively, among control clusters. The terms $\widehat{s}_{{YD}_T}^2$ and $\widehat{s}_{{YD}_C}^2$ are the estimated covariances between $Y$ and $D$ among the treated and control clusters, respectively. Using these estimated variances, the interval in equation \eqref{eq:ci_broad} is a solution to the quadratic equation 
\begin{equation} \label{eq:quadEq}
\{ \tau_0 \mid a \tau_0^2 +2 b \tau_0 + c \leq 0 \}
\end{equation}
where the coefficients $a$, $b$, and $c$ in equation \eqref{eq:quadEq} are 
\begin{align*}
a &= \widehat{\mu}_{D}^2 - z_{1-\alpha/2}^2 \left( \frac{\widehat{s}_{D_T}^2 }{m} + \frac{\widehat{s}_{D_C}^2}{J-m}  \right) \\
b &=  -\left[  \widehat{\mu}_{Y} \widehat{\mu}_{D} - z_{1-\alpha/2}^2\left( \frac{\widehat{s}_{{YD}_T}^2}{m} +\frac{\widehat{s}_{{YD}_C}^2}{J-m} \right) \right] \\
c &=  \widehat{\mu}_{Y}^2 - z_{1-\alpha/2}^2  \left(\frac{\widehat{s}_{Y_T}^2}{m}  + \frac{ \widehat{s}_{Y_C}^2}{J-m}  \right) 
\end{align*}
The $1-\alpha$ confidence interval of $\tau$ based on equation \eqref{eq:quadEq} is simple, requiring a quadratic solver based on coefficients $a$, $b$, and $c$. A notable feature of this confidence interval is that it will produce an infinite confidence intervals when an instrument is weak. An infinite confidence interval is a warning that the data contain little information \citep[p. 185]{Rosenbaum:2002} and is a  necessary property for confidence intervals in IV settings \citep{dufour_some_1997}. Also, this quadratic confidence interval in equation \eqref{eq:quadEq} can never be empty, or mathematically you cannot have $a> 0$ and $(2b)^2 - 4ac < 0$; see the supplementary materials for details. However, this may not be generally true for different variance estimators $Var(T(\tau_0) \mid \mathcal{F},\mathcal{Z})$ which is used to construct confidence intervals of the type in equation \eqref{eq:ci_broad}.

We also propose an exact method for constructing CIs under a more narrow null hypothesis than in equation \eqref{eq:hyp}. Specifically, consider the null $H_{00}$ where the cluster-specific CACE $\tau_j$ are equal to $\tau_0$ for all $j$, i.e. there is cluster-level CACE homogeneity.
\begin{equation} \label{eq:hyp_homoNull}
H_{00}: \tau_j = \tau_0, \quad{} \forall j=1,\ldots, J
\end{equation}
The null $H_{00}$ is a narrower hypothesis than $H_0$ in that the parameters $\mathcal{F}$ that satisfy $H_{00}$ also satisfies $H_0$, but the converse is not true. However, and most importantly, $H_{00}$ is not a strict sharp null in the sense that we still cannot specify all the potential outcomes, say $Y_{ji}^{(z,D_{ji}^{(z)})}$ for any $z\in \{0,1\}$ even under $H_{00}$. Note that under $H_{00}$, the population CACE $\tau$ is equal to $\tau_0$.

A benefit of deriving inferential properties under a restrictive $H_{00}$ is that we can obtain exact, non-parametric, finite-population confidence intervals instead of the asymptotic ones in Proposition \ref{prop:ours}. In particular, the $1-\alpha$ confidence interval derived under $H_{00}$ covers $\tau$ with probability at least $1-\alpha$ in any finite sample.
\begin{proposition} \label{prop:ours_2} Suppose Assumptions (A1)-(A4) hold and the null hypothesis $H_{00}$ in equation \eqref{eq:hyp_homoNull} holds. Then, for any $t \in \reals$ and for any $\tau_0$, the null distribution of $T(\tau_0)$ under $H_{00}$ is the permutation distribution.
\begin{equation} \label{eq:null_2}
P\left(T(\tau_0) \leq t \mid \mathcal{Z}, \mathcal{F} \right) = \frac{\left| \mathbf{z} \in \mathcal{Z} \mid  \frac{1}{m} \sum_{j=1}^{J} z_j A_j(\tau_0) - \frac{1}{J-m} \sum_{j=1}^{J} (1 -z_j) A_{j}(\tau_0) \leq t \right|} {{J \choose m}}
\end{equation}
\end{proposition}
Using the duality between testing and confidence intervals, one can construct a $1-\alpha$ confidence interval for $\tau$ by finding values of $\tau_0$ for which $H_{00}$ is accepted at the $\alpha$ level. For example, let $q_{1-\alpha/2}$ be the $1-\alpha/2$ quantile of the null distribution in equation \eqref{eq:null_2}. Then, a two-sided $1-\alpha$ confidence interval of $\tau$ based on Proposition \ref{prop:ours_2} is the following set of $\tau_0$
\begin{equation} \label{eq:ci_exact}
\left\{\tau_0 : P\left( |T(\tau_0)| \leq q_{1-\alpha/2} \mid \mathcal{Z}, \mathcal{F} \right) \right\}
\end{equation}
where we assumed, for simplicity, that $T(\tau_0)$ is symmetric around $0$; if the latter does not hold, we can take the union of two one-sided confidence intervals to form the desired two-sided confidence interval. Unlike the $1-\alpha$ confidence interval in equation \eqref{eq:ci_exact}, the $1-\alpha$ confidence interval in equation \eqref{eq:ci_exact} is exact in finite-samples, not relying on any asymptotic arguments or estimates of the variance. Also, like the confidence interval in equation \eqref{eq:ci_broad}, if the confidence interval in equation \eqref{eq:ci_exact} is infinite, it is a warning that the data contain little information about $\tau_0$.  Also, if the confidence interval in equation \eqref{eq:ci_exact} is empty, it suggests either that the instrumental variables assumptions \ref{a2}-\ref{a4} are not met or, most likely, that the CACE homogeneity assumption $H_{00}$ is not met since technically speaking, the permutation null distribution in equation \eqref{eq:ci_exact} remains valid even if assumptions \ref{a2}-\ref{a4} fails to hold; see the supplementary materials for technical details. 

A caveat of the interval in equation \eqref{eq:ci_exact} is that one has to compute $q_{1-\alpha/2}$, which depending on the size of $J$ and $m$, may be computationally burdensome. In practice, we recommend using the confidence interval in \eqref{eq:ci_exact} if $J$ is small, roughly under $40$ depending on one's computational power. Next, we evaluate our proposed method against current practice using a simulation study.

\section{Simulation Study} 

We now conduct a simulation study to compare the performance of current methods of estimation and inference to our almost exact approach. As we outlined above, extant methods may not recover the target causal estimand when complier effects vary with cluster size. In the following simulations, we focus on how performance of the various methods change as complier effects vary with cluster sizes. First, we describe general aspects of the simulation. We designed the simulations to closely replicate the conditions of the smartcard intervention.

 Like that study, we assume noncompliance is one-sided. Next, $\pi$ is the probability a unit within a cluster is a complier. Let $P_i$ denote the compliance class, which only includes compliers and never-takers in the one-sided compliance setting, since always-takers and defiers are excluded by the design. We designate a unit as complier, $P_i = co$, by random sampling units from each cluster with probability $\pi$ and $1 - \pi$. We generated outcomes using a linear mixed model (LMM) of the form 
$$
Y_{ji} = \alpha + \tau I(P_i = co) Z_j + \beta n_i + \gamma Z_j n_i + c_i + \epsilon_{ij}.
$$ 
where $co$ indicates that a unit is a complier. Under this model, $Z_{j} = 1$ if cluster $i$ was assigned to treatment and $Z_j = 0$ if the cluster was assigned to control, so $\tau$ is a measure of the individual level treatment effect if the unit is a complier. In the model, $c_i$ is a cluster-level random effect, $\beta$ is a specific cluster-level effect that varies with the size of the cluster, and $\gamma$ allows the treatment effect to vary by cluster size. In the study, $\gamma$ is the key parameter that we vary in the simulations. In the simulations, we used a $t$-distribution with five degrees of freedom for the error terms $c_i$ and $\epsilon_{ij}$.

For this data generating process, we define the intraclass correlation, $\lambda$, as $\Var(c_i)/\{  \Var(c_i) + \Var(\epsilon_{ij})  \} $ when $c_i$ and $\epsilon_{ij}$ have finite variance. We can adjust the scale of the cluster distribution errors $c_i$, so that we can control the value of $\lambda$ in the simulations. In all the simulations, we set $\lambda$ to 0.28, which is equal to the estimated intraclass correlation in the smartcard data. This is is fairly large ICC value. \citet{hedges2007intraclass} using ICC estimates from clustered randomized experiments in education find that $\lambda$'s range from 0.07 to 0.31, with an average value of 0.17. \citet{small_randomization_2008} note that $\lambda$ values in the range of .002 to 0.03 are more typical in public health interventions that target clusters such as hospitals, clinics or villages. 

We repeated the simulation for differing numbers of clusters. We used cluster sizes of 20, 30, 50, 80, 100, and 200. In most CRTs with noncompliance, the number of units per cluster and the compliance rate tend to vary from cluster to cluster. In our simulations, we used the units per cluster and the compliance rates from the smartcard data. We did this by sampling from the actual cluster sizes and compliance rates in the data. That is, when the number of clusters is 50, we took a random sample of 50 cluster sizes and compliance rates from the smartcard data. This allowed us to vary cluster sizes and compliance rates in a fashion that mimics the typical data structure in a CRT with noncompliance.

As we noted above, the key parameter in the simulation is $\gamma$. As such, we conducted three different sets of simulations. In the first, we set $\gamma = 0$, which implies that complier effects do not vary with cluster size.  In the second and third set of simulations, we set $\gamma$ to -0.03 and 0.03. For these two simulations, the complier effects are negatively and positively correlated with cluster sizes respectively.

Table~\ref{tab:sim1} contains the results for all three simulation studies with respect to bias. 
 First, we observe that when complier effects do not vary with cluster sizes all the methods perform the same in terms of bias. Here, even with as few as 20 clusters all three methods recover the true complier average causal effect. In this setting, the bias never exceeds two percent. However, if complier average causal effects are inversely correlated with cluster sizes, the estimates based on cluster level averages underestimates the true complier causal effect by 29 to 34\%. Critically, this bias does not vary with the number of clusters. Notably, in this setting TSLS and the generalized effect ratio perform equally well. Finally, when true complier causal effects are positive correlated, the cluster level averages again under-estimate the true effect by a more modest 6\%. Again TSLS performs as well as the generalized effect ratio. In more limited simulation work, we found the performance of TSLS tends suffer with a greater spread in cluster sizes.

\begin{table}[htbp]
\begin{center}
\begin{threeparttable}
\caption{Bias in Estimators Three Different IV Methods for CRTs}
\label{tab:sim1}
\begin{tabular}{lccccc}
\toprule
Number of & Gen. Effect Ratio & Cluster-level Averages & TSLS \\ 
Clusters &  Bias & Bias & Bias\\
\midrule
\multicolumn{4}{c}{Constant Complier Effects} \\
\midrule
 20 & 1.01 & 1.00 & 1.01 \\ 
   30 & 1.02 & 1.02 & 1.01 \\ 
   50 & 0.99 & 1.00 & 0.99 \\ 
   80 & 1.01 & 1.00 & 1.00 \\ 
  100 & 1.01 & 1.00 & 1.00 \\ 
  200 & 1.00 & 1.00 & 1.00 \\ 
\midrule
\multicolumn{4}{c}{Nonconstant Complier Effects: Negative Correlation} \\
\midrule
   20 & 1.03 & 0.67 & 1.01 \\ 
   30 & 1.05 & 0.71 & 1.01 \\ 
   50 & 1.03 & 0.66 & 1.00 \\ 
   80 & 1.01 & 0.67 & 1.00 \\ 
  100 & 0.97 & 0.66 & 0.99 \\ 
  200 & 0.99 & 0.66 & 1.00 \\  
\midrule
\multicolumn{4}{c}{Nonconstant Complier Effects: Positive Correlation} \\
\midrule
  20 & 0.99 & 0.94 & 1.00 \\ 
   30 & 0.99 & 0.93 & 0.99 \\ 
   50 & 0.99 & 0.93 & 0.99 \\ 
   80 & 1.00 & 0.94 & 0.99 \\ 
  100 & 1.01 & 0.94 & 1.00 \\ 
  200 & 1.00 & 0.94 & 1.00 \\    
\bottomrule
\end{tabular}
\begin{tablenotes}[para]
\small{Note: Cell entries are ratio of average estimate to true effect size.}
\end{tablenotes}
\end{threeparttable}
\end{center}
\end{table}

Next, we focus on inferential properties. In Table~\ref{tab:sim2} we record coverage rates for 95\% confidence intervals. When complier effect are constant, the almost exact method slightly undercovers when sample sizes are 20 or 30 clusters, but nominal coverage is at 95\% for larger sample sizes. Curiously the confidence intervals from the cluster-level averages tend be too narrow, while those based on TSLS were very close to the nominal coverage rate for all sample sizes. When complier effects are negatively correlated with cluster size, confidence intervals from cluster-level averages may undercover by a substantial margin. Here, TSLS overcovers. When complier effects are positively correlated with cluster size, both the almost exact method and TSLS cover at the nominal rate. For cluster-level averages, we find the confidence intervals do not always cover.

\begin{table}[htbp]
\begin{center}
\begin{threeparttable}
\caption{Confidence Interval Coverage for Three Different CRT IV Methods}
\label{tab:sim2}
\begin{tabular}{lccccc}
\toprule
Number of Clusters & Gen. Effect Ratio & Cluster-level Averages & TSLS \\ 
\midrule
\multicolumn{4}{c}{Constant Complier Effects}  \\
\midrule
 20 & 0.94 & 0.97 & 0.96 \\ 
   30 & 0.93 & 0.97 & 0.95 \\ 
   50 & 0.95 & 0.96 & 0.95 \\ 
   80 & 0.96 & 0.97 & 0.96 \\ 
  100 & 0.95 & 0.96 & 0.95 \\ 
  200 & 0.95 & 0.95 & 0.95 \\ 
\midrule
\multicolumn{4}{c}{Non-constant Complier Effects: Negative Correlation} \\
\midrule
   20 & 0.95 & 0.95 & 0.98 \\ 
   30 & 0.96 & 0.93 & 0.98 \\ 
   50 & 0.96 & 0.86 & 0.99 \\ 
   80 & 0.97 & 0.78 & 0.97 \\ 
  100 & 0.96 & 0.69 & 0.99 \\ 
  200 & 0.96 & 0.42 & 0.98 \\ 
\midrule
\multicolumn{4}{c}{Non-constant Complier Effects: Positive Correlation} \\
\midrule
   20 & 0.93 & 0.95 & 0.93 \\ 
   30 & 0.95 & 0.95 & 0.95 \\ 
   50 & 0.96 & 0.92 & 0.97 \\ 
   80 & 0.94 & 0.88 & 0.95 \\ 
  100 & 0.95 & 0.88 & 0.96 \\ 
  200 & 0.95 & 0.72 & 0.96 \\  
\bottomrule
\end{tabular}
\begin{tablenotes}[para]
\small{Note: Cell entries are ratio of average estimate to true effect size.}
\end{tablenotes}
\end{threeparttable}
\end{center}
\end{table}

Table~\ref{tab:sim3} contains the average length of the 95\% confidence intervals for all three methods. The almost exact method also produces a substantially longer confidence intervals than the other two methods. For example when there are 30 clusters the almost exact confidence intervals are nearly twice as long as those based on TSLS. Even with 200 clusters, the almost exact confidence interval tends to be about 40\% longer than those based on asymptotic approximations. Moreover, we found that when there were only 20 clusters, the almost exact method produces an infinite confidence interval approximately one percent of the time. Thus in some instances, the almost exact method provides a clear warning that the instrument is weak.

\begin{table}[htbp]
\begin{center}
\begin{threeparttable}
\caption{Confidence Interval Length of Three Different CRT IV Methods}
\label{tab:sim3}
\begin{tabular}{lccccc}
\toprule
Number of Clusters & Gen. Effect Ratio & Cluster-level Averages & TSLS \\ 
\midrule
\multicolumn{4}{c}{Constant Complier Effects}  \\
\midrule
 20 & 9.23 & 1.10 & 0.96 \\ 
   30 & 1.59 & 0.84 & 0.76 \\ 
   50 & 1.04 & 0.62 & 0.57 \\ 
   80 & 0.78 & 0.49 & 0.45 \\ 
  100 & 0.68 & 0.43 & 0.40 \\ 
  200 & 0.47 & 0.30 & 0.28 \\
\midrule
\multicolumn{4}{c}{Nonconstant Complier Effects: Negative Correlation} \\
\midrule
   20 & 3.30 & 1.03 & 0.95 \\ 
   30 & 1.40 & 0.82 & 0.76 \\ 
   50 & 0.95 & 0.63 & 0.58 \\ 
   80 & 0.69 & 0.48 & 0.45 \\ 
  100 & 0.60 & 0.43 & 0.40 \\ 
  200 & 0.42 & 0.30 & 0.29 \\ 
\midrule
\multicolumn{4}{c}{Nonconstant Complier Effects: Positive Correlation} \\
\midrule
   20 & 4.55 & 1.62 & 1.29 \\ 
   30 & 1.98 & 1.27 & 1.07 \\ 
   50 & 1.38 & 0.95 & 0.85 \\ 
   80 & 1.01 & 0.72 & 0.66 \\ 
  100 & 0.88 & 0.64 & 0.59 \\ 
  200 & 0.62 & 0.45 & 0.42 \\ 
\bottomrule
\end{tabular}
\begin{tablenotes}[para]
\small{Note: Cell entries are ratio of average estimate to true effect size.}
\end{tablenotes}
\end{threeparttable}
\end{center}
\end{table}

In sum, the simulations reveal two important results. First, the estimator based on cluster-level averages uniformly performed the worst once complier effects were correlated with cluster size. Second, contrary to our analytic results, TSLS generally performed well. That implies that it is more robust than using cluster-level averages. However, in a simulation it is difficult to mimic what may be the true relationship between cluster size and complier effects.

\section{Analysis of The Smartcard Intervention} \label{sec:data}

Finally, we turn to an analysis of the data from the smartcard intervention. As we noted above, the intervention was designed to reduce inefficiency in the payment of welfare benefits. In the original study, data were collected on two different aspects of the payment experience.  First, they measured the time, in days, between when the work was completed and the payment was collected. Second, they measured the time in minutes for recipients to collect payments on the day they were disbursed. These data were collected using surveys in each village after the smartcard intervention was in effect. The original study only focused on intention to treat estimates. Here, we report IV estimates using both the two extant methods, and the almost exact method. 

First, we review some basic descriptive statistics. As we noted above, a total of 157 villages participated in the study. Of these, 112 were randomly assigned to treatment. Across the 157 villages, 6,891 people were surveyed, such that cluster sizes varied from six to 85. The key issue we highlighted above is whether complier effects vary with cluster size. Given that the intervention was implemented across a large geographic area and that village sizes varied significantly, there is little reason to believe effects are constant across clusters.

\begin{table}[htbp]
\centering
\caption{Analysis of Smartcard Payment System - Complier Average Causal Effects}
\label{tab.out}
\begin{tabular}{lcccc}
\toprule 
&Gen. Effect Ratio & Cluster-level Averages & TSLS \\ 
\midrule
Payment Delay \\
\midrule
Point Estimate  & -2.72 & -9.14 & -10.15  \\ 
95\% Confidence Interval  & [-21.87,13.56]& [-19.16,0.88]& [-21.17,0.87] \\ 
\midrule
Time to Collect Payment\\
\midrule
Point Estimate  & -27.03 & -57.98 & -60.06  \\ 
95\% Confidence Interval  & [-82.26,26.36]& [-98.61,-17.35]& [-100.05,-20.06] \\ 
\bottomrule
\end{tabular}
\end{table}

Table~\ref{tab.out} contains the results from the three methods. For both outcomes, cluster-level averages and TSLS produce very similar results for both the point estimates and the confidence intervals. However, the estimates based on the generalized effect ratio are much smaller in magnitude and the confidence intervals are much wider. Thus there are striking similarities and contrasts between the empirical results and the simulations. First, as was the case in the simulations, confidence intervals based on the almost exact method are notably longer. However, in the simulations, the generalized effect ratio and TSLS produced very similar point estimates. Here, we find that the generalized effect ratio differs from the other two methods which are very similar. In fact, in the data, we find the pattern suggested by the analytic results: the extant methods substantially overestimate the CACE. What might explain the difference between the simulations and the empirical results?  Most likely, the pattern of complier effects varies in a more complex fashion than was true in the simulations. Thus the empirics are more closely aligned with the analytic results which suggest that TSLS does not recover the causal estimand of interest in settings of this type.

\begin{figure}
\centering
\begin{subfigure}[b]{.47\textwidth}
  \includegraphics[width=\textwidth]{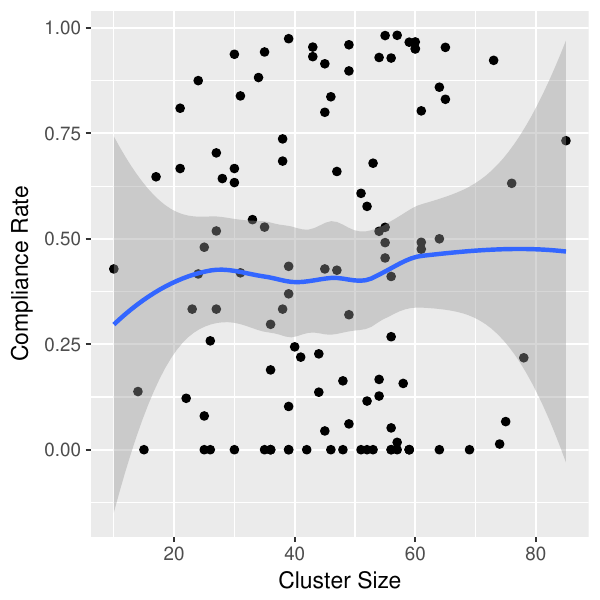}
  \label{fig:c1}
\end{subfigure}
\begin{subfigure}[b]{.47\textwidth}
  \includegraphics[width=\textwidth]{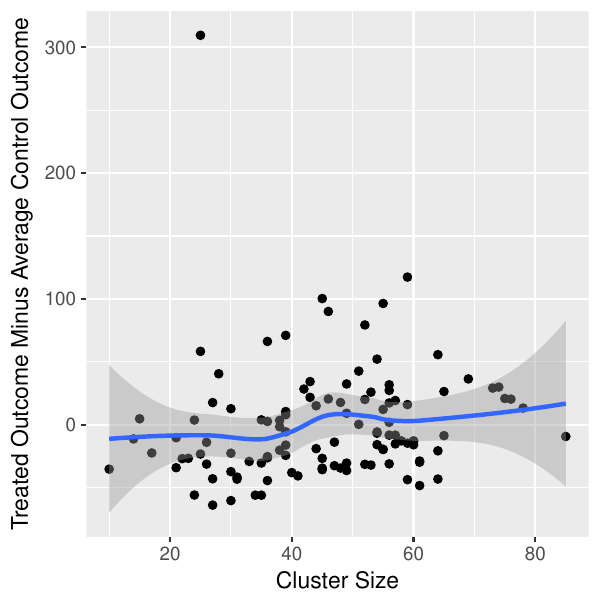}
  \label{fig:out}
\end{subfigure}

\caption{Scatterplots between cluster size, compliance rates and treated cluster outcome deviated from average control outcome for time to collect outcome.}
\label{fig:plots}
\end{figure}

To explore this possibility, we plotted the compliance rate against the cluster size for each of the 157 villages included in the study. This plot is contained in the first panel of Figure~\ref{fig:plots}. There appears to be little relationship between cluster size and compliance. Next, we calculated the average outcome for each treated cluster and subtracted the average control outcome from each of the treated averages. We plotted these adjusted treated outcomes against cluster size. This plot is in the second panel of Figure~\ref{fig:plots}. We find a weakly positive relationship between outcomes and cluster size.

\section{Discussion}

Assessment of policy interventions is often done using cluster randomized trials. This study design has two key advantages. First, it relies on randomization to remove both hidden and overt bias in the estimation of treatment effect. Second, it allows for natural patterns of interaction within clusters to reduce the likelihood of spillovers from treatment to control. The use of smartcards to deliver welfare payments in India follows such a template.

We demonstrated through analytic results, simulations, and empirics that extant methods may not recover the target causal estimand when investigators are interested in the causal effect among those that actually complied with the treatment. These methods impose the strong assumption that complier effects do not vary with cluster size or that the cluster size is identical. Additionally, both methods rely on asymptotic assumptions for interval estimation where the compliance rate is assumed to be uniformly high. Here, we have developed methods that allow for consistent estimates of the CACE when complier effects or cluster size vary. We introduce an approach to inference that provides correct confidence intervals when instruments are weak. Our method better approximates finite samples and has a closed-form solution for point estimation and inference.

\clearpage
\singlespacing
\bibliographystyle{asa}
\bibliography{keele_revised2,ref}

\begin{thebibliography}{46}
\newcommand{\enquote}[1]{``#1''}
\expandafter\ifx\csname natexlab\endcsname\relax\def\natexlab#1{#1}\fi

\bibitem[{Acemoglu and Angrist(2000)}]{acemoglu2000large}
Acemoglu, D. and Angrist, J. (2000), \enquote{How large are human-capital
  externalities? Evidence from compulsory schooling laws,} \textit{NBER
  macroeconomics annual}, 15, 9--59.

\bibitem[{Angrist et~al.(1996)Angrist, Imbens, and Rubin}]{Angrist:1996}
Angrist, J.~D., Imbens, G.~W., and Rubin, D.~B. (1996), \enquote{Identification
  of Causal Effects Using Instrumental Variables,} \textit{Journal of the
  American Statistical Association}, 91, 444--455.

\bibitem[{Angrist and Pischke(2009)}]{Angrist:2009}
Angrist, J.~D. and Pischke, {\"J}.-S. (2009), \textit{Mostly Harmless
  Econometrics}, Princeton, NJ: Princeton University Press.

\bibitem[{Aronow et~al.(2014)Aronow, Green, Lee, et~al.}]{aronow_sharp_2014}
Aronow, P.~M., Green, D.~P., Lee, D.~K., et~al. (2014), \enquote{Sharp bounds
  on the variance in randomized experiments,} \textit{The Annals of
  Statistics}, 42, 850--871.

\bibitem[{Baiocchi et~al.(2014)Baiocchi, Cheng, and
  Small}]{baiocchi_instrumental_2014}
Baiocchi, M., Cheng, J., and Small, D.~S. (2014), \enquote{Instrumental
  variable methods for causal inference,} \textit{Statistics in medicine}, 33,
  2297--2340.

\bibitem[{Cameron and Miller(2015)}]{cameron2015practitioner}
Cameron, A.~C. and Miller, D.~L. (2015), \enquote{A practitioner's guide to
  cluster-robust inference,} \textit{Journal of Human Resources}, 50, 317--372.

\bibitem[{Donner and Klar(2000)}]{donner2000design}
Donner, A. and Klar, N. (2000), \textit{Design and Analysis of Cluster
  Randomization Trials in Health Research}, New York: Wiley.

\bibitem[{Dufour(1997)}]{dufour_some_1997}
Dufour, J.-M. (1997), \enquote{Some impossibility theorems in econometrics with
  applications to structural and dynamic models,} \textit{Econometrica},
  1365--1387.

\bibitem[{Dutta et~al.(2010)Dutta, Howes, and Murgai}]{dutta2010small}
Dutta, P., Howes, S., and Murgai, R. (2010), \enquote{Small but effective:
  India's targeted unconditional cash transfers,} \textit{Economic and
  Political Weekly}, 45, 63--70.

\bibitem[{Fisher(1935)}]{Fisher:1935}
Fisher, R.~A. (1935), \textit{The Design of Experiments}, London: Oliver and
  Boyd.

\bibitem[{Frangakis et~al.(2002)Frangakis, Rubin, and
  Zhou}]{frangakis_clustered_2002}
Frangakis, C.~E., Rubin, D.~B., and Zhou, X.-H. (2002), \enquote{Clustered
  encouragement designs with individual noncompliance: Bayesian inference with
  randomization, and application to advance directive forms,}
  \textit{Biostatistics}, 3, 147--164.

\bibitem[{Gerber et~al.(2008)Gerber, Green, and Larimer}]{Gerber:2008}
Gerber, A.~S., Green, D.~P., and Larimer, C.~W. (2008), \enquote{Social
  Pressure and Voter Turnout: Evidence From a Large-Scale Field Experiment,}
  \textit{American Political Science Review}, 102, 33--48.

\bibitem[{Green et~al.(2013)Green, McGrath, and Aronow}]{Green:2013}
Green, D.~P., McGrath, M.~C., and Aronow, P.~M. (2013), \enquote{Field
  Experiments and the Study of Voter Turnout,} \textit{Journal of Elections,
  Public Opinion, and Parties}, 23, 27--48.

\bibitem[{H{\'a}jek(1960)}]{Hajek:1960}
H{\'a}jek, J. (1960), \enquote{Limiting distributions in simple random sampling
  from a finite population,} \textit{Publications of the Mathematics Institute
  of the Hungarian Academy of Science}, 5, 361--374.

\bibitem[{Hansen and Bowers(2009)}]{Hansen:2008b}
Hansen, B.~B. and Bowers, J. (2009), \enquote{Attributing Effects to A
  Clustered Randomized Get-Out-The-Vote Campaign,} \textit{Journal of the
  American Statistical Association}, 104, 873--885.

\bibitem[{Hayes et~al.(2014)Hayes, Ayles, Beyers, Sabapathy, Floyd, Shanaube,
  Bock, Griffith, Moore, Watson-Jones, et~al.}]{hayes2014hptn}
Hayes, R., Ayles, H., Beyers, N., Sabapathy, K., Floyd, S., Shanaube, K., Bock,
  P., Griffith, S., Moore, A., Watson-Jones, D., et~al. (2014), \enquote{HPTN
  071 (PopART): Rationale and design of a cluster-randomised trial of the
  population impact of an HIV combination prevention intervention including
  universal testing and treatment--a study protocol for a cluster randomised
  trial,} \textit{Trials}, 15, 57.

\bibitem[{Hayes and Moulton(2009)}]{hayes2009cluster}
Hayes, R. and Moulton, L. (2009), \textit{Cluster Randomised Trials}, Chapman\&
  Hall/CRC.

\bibitem[{Hedges and Hedberg(2007)}]{hedges2007intraclass}
Hedges, L.~V. and Hedberg, E.~C. (2007), \enquote{Intraclass correlation values
  for planning group-randomized trials in education,} \textit{Educational
  Evaluation and Policy Analysis}, 29, 60--87.

\bibitem[{Hern{\'a}n and Robins(2006)}]{hernan_instruments_2006}
Hern{\'a}n, M.~A. and Robins, J.~M. (2006), \enquote{Instruments for causal
  inference: an epidemiologist's dream?} \textit{Epidemiology}, 17, 360--372.

\bibitem[{Hodges and Lehmann(1963)}]{hodges_estimates_1963}
Hodges, J.~L. and Lehmann, E.~L. (1963), \enquote{Estimates of Location Based
  on Rank Tests,} \textit{Ann. Math. Statist.}, 34, 598--611.

\bibitem[{Imai et~al.(2009)Imai, King, Nall, et~al.}]{imai_essential_2009}
Imai, K., King, G., Nall, C., et~al. (2009), \enquote{The essential role of
  pair matching in cluster-randomized experiments, with application to the
  Mexican universal health insurance evaluation,} \textit{Statistical Science},
  24, 29--53.

\bibitem[{Imbens and Wooldridge(2008)}]{imbens2008recent}
Imbens, G.~M. and Wooldridge, J.~M. (2008), \enquote{Recent developments in the
  econometrics of program evaluation,} \textit{Journal of Economic Literature},
  47, 5--86.

\bibitem[{Imbens and Angrist(1994)}]{Imbens:1994}
Imbens, G.~W. and Angrist, J.~D. (1994), \enquote{Identification and Estimation
  of Local Average Treatment Effects,} \textit{Econometrica}, 62, 467--476.

\bibitem[{Imbens and Rubin(2015)}]{imbens_causal_2015}
Imbens, G.~W. and Rubin, D.~B. (2015), \textit{Causal Inference in Statistics,
  Social, and Biomedical Sciences}, Cambridge University Press.

\bibitem[{Jo et~al.(2008{\natexlab{a}})Jo, Asparouhov, and
  Muth{\'e}n}]{jo_intention_2008}
Jo, B., Asparouhov, T., and Muth{\'e}n, B.~O. (2008{\natexlab{a}}),
  \enquote{Intention-to-treat analysis in cluster randomized trials with
  noncompliance,} \textit{Statistics in medicine}, 27, 5565--5577.

\bibitem[{Jo et~al.(2008{\natexlab{b}})Jo, Asparouhov, Muth{\'e}n, Ialongo, and
  Brown}]{jo2008cluster}
Jo, B., Asparouhov, T., Muth{\'e}n, B.~O., Ialongo, N.~S., and Brown, C.~H.
  (2008{\natexlab{b}}), \enquote{Cluster randomized trials with treatment
  noncompliance.} \textit{Psychological methods}, 13, 1.

\bibitem[{Kang et~al.(2018)Kang, Peck, and Keele}]{Keele:2017fiv}
Kang, H., Peck, L., and Keele, L. (2018), \enquote{Inference for Instrumental
  Variables: A Randomization Inference Approach,} \textit{Journal of The Royal
  Statistical Society, Series A}, In press.

\bibitem[{Lehmann(2004)}]{Lehmann:2004}
Lehmann, E.~L. (2004), \textit{Elements of Large-Sample Theory}, Springer.

\bibitem[{Li and Ding(2017)}]{li_general_2017}
Li, X. and Ding, P. (2017), \enquote{General forms of finite population central
  limit theorems with applications to causal inference,} \textit{Journal of the
  American Statistical Association}.

\bibitem[{Liang and Zeger(1986)}]{liang1986longitudinal}
Liang, K.-Y. and Zeger, S.~L. (1986), \enquote{Longitudinal data analysis using
  generalized linear models,} \textit{Biometrika}, 13--22.

\bibitem[{Middleton and Aronow(2015)}]{middleton_unbiased_2015}
Middleton, J.~A. and Aronow, P.~M. (2015), \enquote{Unbiased estimation of the
  average treatment effect in cluster-randomized experiments,}
  \textit{Statistics, Politics and Policy}, 6, 39--75.

\bibitem[{Muralidharan et~al.(2016)Muralidharan, Niehaus, and
  Sukhtankar}]{muralidharan2016building}
Muralidharan, K., Niehaus, P., and Sukhtankar, S. (2016), \enquote{Building
  state capacity: Evidence from biometric smartcards in India,}
  \textit{American Economic Review}, 106, 2895--2929.

\bibitem[{Neyman(1923)}]{Neyman:1923a}
Neyman, J. (1923), \enquote{On the Application of Probability Theory to
  Agricultural Experiments. Essay on Principles. Section 9.}
  \textit{Statistical Science}, 5, 465--472. Trans. Dorota M. Dabrowska and
  Terence P. Speed (1990).

\bibitem[{Pustejovsky and Tipton(2018)}]{pustejovskycluster2018}
Pustejovsky, J.~E. and Tipton, E. (2018), \enquote{Small-Sample Methods for
  Cluster-Robust Variance Estimation and Hypothesis Testing in Fixed Effects
  Models,} \textit{Journal of Business \& Economic Statistics}, 36, 672--683.

\bibitem[{Robins(1988)}]{robins_confidence_1988}
Robins, J.~M. (1988), \enquote{Confidence intervals for causal parameters,}
  \textit{Statistics in Medicine}, 7, 773--785.

\bibitem[{Rosenbaum(2002)}]{Rosenbaum:2002}
Rosenbaum, P.~R. (2002), \textit{Observational Studies}, New York, NY:
  Springer, 2nd ed.

\bibitem[{Rubin(1974)}]{rubin_estimating_1974}
Rubin, D.~B. (1974), \enquote{Estimating causal effects of treatments in
  randomized and nonrandomized studies.} \textit{Journal of Educational
  Psychology}, 66, 688.

\bibitem[{Rubin(1986)}]{Rubin:1986}
--- (1986), \enquote{Which Ifs Have Causal Answers,} \textit{Journal of the
  American Statistical Association}, 81, 961--962.

\bibitem[{Schochet(2013)}]{schochet2013estimators}
Schochet, P.~Z. (2013), \enquote{Estimators for clustered education RCTs using
  the Neyman model for causal inference,} \textit{Journal of Educational and
  Behavioral Statistics}, 38, 219--238.

\bibitem[{Schochet and Chiang(2011)}]{schochet_estimation_2011}
Schochet, P.~Z. and Chiang, H.~S. (2011), \enquote{Estimation and
  identification of the complier average causal effect parameter in education
  RCTs,} \textit{Journal of Educational and Behavioral Statistics}, 36,
  307--345.

\bibitem[{Small et~al.(2006)Small, Ten~Have, Joffe, and
  Cheng}]{small2006random}
Small, D.~S., Ten~Have, T.~R., Joffe, M.~M., and Cheng, J. (2006),
  \enquote{Random effects logistic models for analysing efficacy of a
  longitudinal randomized treatment with non-adherence,} \textit{Statistics in
  Medicine}, 25, 1981--2007.

\bibitem[{Small et~al.(2008)Small, Ten~Have, and
  Rosenbaum}]{small_randomization_2008}
Small, D.~S., Ten~Have, T.~R., and Rosenbaum, P.~R. (2008),
  \enquote{Randomization inference in a group--randomized trial of treatments
  for depression: covariate adjustment, noncompliance, and quantile effects,}
  \textit{Journal of the American Statistical Association}, 103, 271--279.

\bibitem[{Solomon et~al.(2015)Solomon, Mehta, Srikrishnan, Vasudevan, Mcfall,
  Balakrishnan, Anand, Nandagopal, Ogburn, Laeyendecker,
  et~al.}]{solomon2015high}
Solomon, S.~S., Mehta, S.~H., Srikrishnan, A.~K., Vasudevan, C.~K., Mcfall,
  A.~M., Balakrishnan, P., Anand, S., Nandagopal, P., Ogburn, E., Laeyendecker,
  O., et~al. (2015), \enquote{High HIV prevalence and incidence among men who
  have sex with men (MSM) across 12 cities in India,} \textit{AIDS (London,
  England)}, 29, 723.

\bibitem[{Stock et~al.(2002)Stock, Wright, and Yogo}]{stock2002survey}
Stock, J.~H., Wright, J.~H., and Yogo, M. (2002), \enquote{A survey of weak
  instruments and weak identification in generalized method of moments,}
  \textit{Journal of Business \& Economic Statistics}, 20, 518--529.

\bibitem[{Wald(1940)}]{wald1940}
Wald, A. (1940), \enquote{The Fitting of Straight Lines if Both Variables are
  Subject to Error,} \textit{The Annals of Mathematical Statistics}, 11,
  284--300.

\bibitem[{Wooldridge(2010)}]{wooldridge2010econometric}
Wooldridge, J.~M. (2010), \textit{Econometric analysis of cross section and
  panel data}, MIT press.

\end{thebibliography}

\end{document}